\documentclass[onecolumn,12pt]{IEEEtran}

\usepackage{cite,citesort}
\usepackage{graphicx,color,psfrag}
\usepackage[tbtags]{amsmath}
\usepackage{times}
\usepackage{latexsym}
\usepackage{bm}
\usepackage{amssymb}
\usepackage{stfloats}
\usepackage{cases}
\usepackage{array}
\usepackage{setspace}
\usepackage{fancyhdr}

\newtheorem{theorem}{Theorem}
\newtheorem{lemma}{Lemma}

\newtheorem{conjecture}{Conjecture}
\newtheorem{definition}{Definition}
\newcommand{\Per}{{\rm Per}}
\newcommand{\diag}{{\rm diag }}

\begin{document}

\title{Statistical Eigenmode Transmission over Jointly-Correlated MIMO Channels}
\author{
\vspace*{0.5cm}
\begin{center}
\authorblockN{Xiqi Gao$^{*}$,
              Bin Jiang$^*$, Xiao Li$^*$,
              Alex B. Gershman$^\dagger$,
              and Matthew R. McKay$^\ddagger$}
\vspace*{0.5cm}
\thanks{ $^*$ National Mobile Communications
Research Laboratory, Southeast University, Nanjing, China}
\thanks{ $^\dagger$ Department of Communication Systems, Technische
Universit\"at Darmstadt, Germany}
\thanks{ $^\ddagger$ Electronic and Computer Engineering Department, Hong Kong
University of Science and Technology, Hong Kong}
\end{center}
}
\maketitle
\vspace*{0.5cm}
\begin{abstract}
We investigate MIMO eigenmode transmission using statistical channel
state information at the transmitter. We consider a general
jointly-correlated MIMO channel model, which does not require
separable spatial correlations at the transmitter and receiver. For
this model, we first derive a closed-form tight upper bound for the
ergodic capacity, which reveals a simple and interesting
relationship in terms of the matrix permanent of the eigenmode
channel coupling matrix and embraces many existing results in the
literature as special cases. Based on this closed-form and tractable
upper bound expression, we then employ convex optimization
techniques to develop low-complexity power allocation solutions
involving only the channel statistics. Necessary and sufficient
optimality conditions are derived, from which we develop an
iterative water-filling algorithm with guaranteed convergence.
Simulations demonstrate the tightness of the capacity upper bound
and the near-optimal performance of the proposed low-complexity
transmitter optimization approach.
\end{abstract}


\newpage

\section{Introduction}

Multiple-input multiple-output (MIMO) wireless systems, equipped
with multiple antennas at both the transmitter and the receiver,
have attracted tremendous interest in recent years as a means of
enabling substantially increased link capacity and reliability
compared with conventional systems
\cite{Telatar,Foschini_limit,Paulraj03,Gershman,Giannakis07,Stuber04}.
The performance of practical MIMO systems is characterized by
various system parameters, such as the average transmit power and
the transmit-receive antenna configurations, as well as various
channel phenomena such as spatial correlation, line-of-sight
components, thermal noise, interference, and Doppler effects due to
mobility. Each of these factors has an impact on the MIMO channel
capacity.

In realistic environments, where the channel characteristics may
vary significantly over time, substantial MIMO capacity benefits can
be obtained by tracking the states of the fading channels, and using
this channel state information (CSI) to optimally adapt the MIMO
transceiver parameters. However, such closed-loop MIMO strategies
require both the transmitter and receiver to acquire some form of
CSI. Whilst it is reasonable to assume that the instantaneous CSI
can be obtained accurately at the receiver through channel
estimation, whether or not this information can be obtained at the
transmitter depends highly on the application scenario. For example,
for fixed or low mobility applications, the channel conditions vary
relatively slowly, in which case instantaneous CSI can be fed to the
transmitter via well-designed feedback channels in frequency
division duplex (FDD) systems, or using the reciprocity of uplink
and downlink in time division duplex (TDD) systems. However, as the
mobility and hence the fading rate increases, obtaining accurate
instantaneous CSI at the transmitter becomes much more difficult.
For such a scenario, it is reasonable to exploit \emph{statistical}
CSI at the transmitter. The motivation for this approach stems from
the fact that the channel statistics vary over much larger time
scales than the instantaneous channel gains, and the uplink and
downlink statistics are usually reciprocal in both FDD and TDD
systems\cite{Barriac2004,Barriac2006}. Therefore, the statistical
information can be easily obtained by exploiting reciprocity, or by
employing feedback channels with significantly lower bandwidth
compared with instantaneous CSI feedback systems. In addition,
transceiver designs based on statistical information are typically
more robust to imperfections, such as delays, in the feedback
channel.

Capacity analysis and transceiver designs using the statistical CSI
at the transmitter are highly dependant on the channel modeling. The
most common approach has been to adopt the popular Kronecker model
\cite{Shui,Shin,Chiani,Cui_TC,Jin,Mckay,Jafar_comm,Jorswieck,Kiessling_vtc,Kiessling_diss,Boche_sp},
where the correlation between the fading of two distinct antenna
pairs is the product of the corresponding transmit and receive
correlations \cite{Shui,Kermoal}. The primary advantage of this
separable model is that it is analytically friendly, however various
measurement campaigns have demonstrated that it can have
deficiencies in practice \cite{Weichselberger,Zhou}. To overcome
these deficiencies, more generalized channel models have been
proposed, including the virtual channel representation
\cite{Sayeed,Veeravalli}, the unitary-independent-unitary (UIU)
model \cite{Tulino2005,Tulino2006}, and Weichselberger's model
\cite{Weichselberger}. In contrast to the Kronecker model, these are
\emph{jointly-correlated} channel models which not only account for
the correlation at both link ends, but also characterize their
mutual dependence.

Under various assumptions on the system configuration and channel
model, several important works have been reported on the transmitter
optimization problem using the statistical CSI at the transmitter in
recent years. In particular, for multi-input and single-output
(MISO) wireless channels with correlated Rayleigh or uncorrelated
Rician entries, it was shown in \cite{Visotsky} that the
capacity-achieving strategy is to send independent data streams in
the directions defined by the dominant eigenvectors of the transmit
correlation matrix. This result was extended to Rayleigh fading MIMO
channels with Kronecker correlation structure in
\cite{Jafar_comm,Jorswieck,Kiessling_vtc,Kiessling_diss}, to
uncorrelated Rician MIMO channels in \cite{Venkatesan,Hosli}, to the
UIU model in \cite{Tulino2005,Tulino2006}, and to the virtual
channel representation in \cite{Veeravalli}. These prior
contributions have also considered the task of optimally allocating
power across the transmit eigendirections (i.e.\ defining the
eigenvalues of the optimal transmit covariance matrix), for
maximizing capacity. However, in most cases, the power allocation
problem has been tackled by optimizing the exact ergodic capacity
expression, and this approach has led to computationally-involved
numerical optimization procedures. For example, see \cite{Vu,Hanlen}
for Kronecker channels, and \cite{Tulino2006} for jointly-correlated
UIU channels. In these contributions, iterative power allocation
approaches were presented which involved numerical averaging over
channel samples for each iteration of the algorithm.

In this paper, we investigate the statistical eigenmode transmission
(SET) over a general jointly-correlated MIMO channel, using the
statistical CSI at the transmitter. Our idea is to first derive a
\emph{closed-form} tight upper bound for the ergodic capacity of the
general jointly-correlated MIMO channel model. This upper bound
expression reveals a simple and interesting relationship in terms of
matrix permanents, and embraces many existing results in the
literature as special cases, such as those presented for Kronecker
channels in \cite{Shin,Kiessling_vtc,Kiessling_diss,Jin,Mckay}.
Based on this closed-form and tractable upper bound expression, we
then employ convex optimization techniques to develop low-complexity
power allocation solutions in terms of only the channel statistics.
We derive necessary and sufficient optimality conditions, and
propose a simple computation algorithm, inspired by the iterative
water-filling techniques presented previously for transmitter
optimization of multiuser systems \cite{Yu,Jindal}, which is shown
to converge within only a few iterations. Numerical simulations
demonstrate the tightness of the capacity upper bound and the
near-optimal performance of the proposed low-complexity transmitter
optimization approach, i.e., suffering negligible loss with respect
to the ergodic capacity of the jointly-correlated MIMO channel.

\subsection{Notation}
The following notation is adopted throughout the paper: Upper
(lower) bold-face letters are used to denote matrices (column
vectors); in some cases, where it is not clear, we will employ
subscripts to emphasize dimensionality. The $N \times N$ identity
matrix is denoted by ${\bf{I}}_N $, the all-zero matrix is denoted
by ${\bf{0}}$, and the all-one matrix is denoted by ${\bf{1}}$. The
superscripts $(\cdot)^{H}$, $(\cdot)^{T}$, and $(\cdot)^{*}$ stand
for the conjugate-transpose, transpose, and conjugate operations,
respectively. We employ $\mathbb{E}\{\cdot\}$ to denote expectation
with respect to all random variables within the brackets, and use
${\bf A}\odot{\bf B}$ to denote the Hadamard product of the two
matrices $\bf A$ and $\bf B$. We use $[{\bf A}]_{kl}$ or the
lower-case representation $a_{k,l}$ to denote the ($k$,$l$)-th entry
of the matrix $\bf A$, and $a_k$ denotes the $k$-th entry of the
column vector $\bf{a}$. The operators ${\rm tr}(\cdot)$,
$\det(\cdot)$, and $\Per(\cdot)$ represent the matrix trace,
determinant, and permanent, respectively, and $\diag(\bf{x})$
denotes a diagonal matrix with $\bf{x}$ along its main diagonal.

We will use $\mathcal {S}_N^{k}$ to denote the set of all size-$k$
permutations of the numbers $\{1,2,\ldots ,N\}$, where $k\leq N$. By
using the notation ${\hat{\alpha}_{k}}\in \mathcal{S}_N^{k}$, we
mean that ${\hat{\alpha}_k}=({\alpha}_1,{\alpha}_2,\ldots,{\alpha}_k
)$, ${\alpha}_i \in \{1,2,\ldots ,N\}$ for $1\leq i\leq k$, and
${\alpha}_i \neq {\alpha}_j$ for $1\leq i,j \leq k$ and $i\neq j $.
We will use $\mathcal {S}_N^{(k)}$ to denote the set of all ordered
length-$k$ subsets of the numbers $\{1,2,\ldots ,N\}$. By the
notation ${\hat{\alpha}}_k\in \mathcal{S}_N^{(k)}$, we mean that
${\hat {\alpha}}_k=({\alpha}_1,{\alpha}_2,\ldots ,{\alpha}_k )$,
${\alpha}_i \in \{1,2,\ldots ,N\}$ for $1\leq i\leq k$, and
${\alpha}_1 < {\alpha}_2 <\ldots < {\alpha}_k$. The cardinalities of
the sets $\mathcal {S}_N ^k $ and $\mathcal {S}_N^{(k)}$ are
$\frac{N!}{(N-k)!}$ and $\frac{N!}{k!(N-k)!}$ respectively.

With ${\hat{\alpha}_{k}}$ and ${\hat{\beta}_{k}}$ defined as above,
we will use ${\bf A}^{\hat{\alpha}_{k}}_{\hat{\beta}_{k}}$ to denote
the sub-matrix of an $M\times N$ matrix $\bf A$ obtained by
selecting the rows and columns indexed by $\hat \alpha _k$ and $\hat
\beta _k$ respectively. $\mathbf{A}^{\hat \alpha_N }$ will denote
the sub-matrix of $\mathbf{A}$ obtained by selecting the rows
indexed by $\hat \alpha _N$ when $M \geq N $, and $\mathbf{A}_{\hat
\beta _M }$ the sub-matrix of $\mathbf{A}$ obtained by selecting the
columns indexed by $\hat \beta _M$ when $M\leq N$. Also, we will use
${\hat{\alpha}^{'}_{k}}$ and ${\hat{\beta}^{'}_{k}}$ to denote the
sequences complementary to ${\hat{\alpha}_{k}}$ and
${\hat{\beta}_{k}}$ in $\{1,2,\ldots, M\}$ and $\{1,2,\ldots , N\}$,
respectively.  As such, ${\bf
A}^{\hat{\alpha}^{'}_{k}}_{\hat{\beta}^{'}_{k}}$ will represent the
sub-matrix of $\bf A$ obtained by deleting the rows and columns
indexed by $\hat \alpha _k$ and $\hat \beta _k$, respectively.
Finally, for notational convenience, we will use ${\bf A}_{i,j}$ to
represent the sub-matrix of $\bf A$ obtained by deleting its $i$-th
row and $j$-th column.

\section{Channel Model and Statistical Eigenmode Transmission}\label{sec:model}
\subsection{Channel Model}
We consider a single-user MIMO link with $N_t$ transmit and $N_r$
receive antennas, operating over a frequency-flat fading channel.
The $N_r$-dimensional complex baseband received signal vector for a
single symbol interval can be written as
\begin{equation}\label{SIGNAL_MODEL}
{\bf{y}} = {\bf{Hx}} + {\bf{n}},
\end{equation}
where $\bf x$ is the $N_t\times 1$ transmitted signal vector, $\bf
H$ is the $N_r \times N_t$ channel matrix with $(i,j)$-th element
representing the complex fading coefficient between the $j$-th
transmit and $i$-th receive antenna, and $\bf n$ is the $N_r\times
1$ zero-mean additive complex Gaussian noise vector with
$\mathbb{E}\left\lbrace {\bf n} {\bf n}^{H}\right\rbrace =\sigma
_{n}^{2} {\bf I}_{N_r}$. It is assumed that $\bf x$ and $\bf H$
satisfy the following power constraints
\begin{align}\label{Constraint on transmit power}
\mathbb{E}\left\{ {\rm tr} \left( {\bf{x}} {\bf{x}}^H \right)
\right\} &= P, \\ \label{Constraint on channel power}
\mathbb{E}\left\{ {\rm tr} \left( {\bf{H}}{\bf{H}}^H \right)
\right\} &= N_t N_r.
\end{align}
We define the transmit signal to noise ratio (SNR) as $\rho =
P/\sigma _{n}^{2}$. If the total transmitted power $P$ is equally
distributed across all transmit antennas, so that
$\mathbb{E}\left\lbrace {\bf x} {\bf x}^H \right\rbrace =\left(
P/N_t \right) {\bf I}_{N_t}$, then $\rho$ also corresponds to the
average SNR per receive antenna.

For the jointly-correlated MIMO channel which we consider throughout
this paper, the channel matrix $\bf H$ is given by
\begin{align}
{\bf {H}}&={\bf{U}}_{r} {\bf \tilde{H}}  {\bf {U}}_{t}^{H}  \nonumber \\
&={\bf{U}}_{r} \left( {\bf D} + {\bf{M}} \odot {{\bf{H}}_{iid}}
\right) {\bf {U}}_{t}^{H}, \label{channel_matrix}
\end{align}
where ${\bf \tilde{H}}={\bf D} + {\bf{M}} \odot {{\bf{H}}_{iid}}$,
${\bf{U}}_{t}$ and ${\bf{U}}_{r}$ are $N_t \times N_t$ and
$N_r\times N_r$ deterministic unitary matrices, ${\bf D}$ is an $N_r
\times N_t$ deterministic matrix with at most one nonzero element in
each row and each column, $\bf M$ is an  $N_r\times N_t$
deterministic matrix with nonnegative elements, and ${\bf{H}}_{iid}
$ is an $N_r\times N_t$ random matrix with elements having zero mean
and independent identical distributions (i.i.d.). Note that we do
not constrain the elements of ${\bf{H}}_{iid} $ to be Gaussian.
Without any loss of generality, we can assume that the nonzero
elements of ${\bf D}$ are real, with indices $\left( i, i \right) $
for $1 \leq i \leq {\rm min} (N_t,N_r)$. Let us define
\begin{equation} \label{Channel coupling matrix}
{\bf \Omega }= \mathbb{E} \{ {\bf \tilde {H}}\odot {\bf \tilde
{H}}^{*} \}.
\end{equation}
It is easy to show that
\begin{equation} \label{Channel coupling matrix expression}
{\bf \Omega } = {\bf D} \odot {\bf D} + {\bf M}
\odot {\bf M},
\end{equation}
and the power constraint (\ref{Constraint on channel power}) can be
rewritten as
\begin{equation}
\sum\limits_{i = 1}^{N_r} \sum\limits_{j = 1}^{N_t} [{\bf \Omega
}]_{ij} = N_t N_r \; .
\end{equation}

From (\ref{channel_matrix}), the transmit and receive correlation
matrices can be expressed as
\begin{equation}\label{Transmit correlation matrix}
{\bf R}_t = \mathbb{E} \{ {\bf H}^{H} {\bf H} \}= {\bf U}_{t} {\bf
\Lambda}_{t} {\bf U}_{t}^{H},
\end{equation}
\begin{equation}
\label{Receive correlation matrix} {\bf R}_r = \mathbb{E} \{ {\bf H}
{\bf H}^{H} \}= {\bf U}_{r} {\bf \Lambda}_{r} {\bf U}_{r}^{H},
\end{equation}
where ${\bf \Lambda}_{t} $ and ${\bf \Lambda}_{r} $ are diagonal
matrices with $ \left[  {\bf \Lambda}_{t} \right] _{ii} =
\sum_{j=1}^{N_r } \left[ \bf \Omega \right] _{ji}$ and $ \left[ {\bf
\Lambda}_{r} \right]  _{ii} = \sum_{j=1}^{N_t } \left[ \bf \Omega
\right]_{ij}$. This implies that in the channel model defined in
(\ref{channel_matrix}), ${\bf U}_t$ and ${\bf U}_r$ are the
eigenvector matrices of the transmit and receive correlation
matrices, respectively. These matrices are characterized by the
transmit and receive antenna configurations. For example, when
uniform linear arrays (ULA) are employed at both the transmitter and
receiver, it is shown in \cite{Sayeed} that the eigenvector matrices
can be set to discrete Fourier transform (DFT) matrices.

The statistics of ${\bf \tilde H}={\bf U}_{r}^{H}{\bf H}{\bf U}_{t}$
characterize realistic propagation environments. From
(\ref{channel_matrix}) and (\ref{Channel coupling matrix
expression}), we have
\begin{align}
{\bf D}&=\mathbb{E} \{ \bf \tilde H  \}, \\
[{\bf M}]_{ij}&=\sqrt{{\rm var}\{[{ \bf \tilde H }]_{ij}\}}
\nonumber
\\&=\sqrt{[{\bf \Omega}]_{ij}-{[\bf D]}_{ij}^2},
\end{align}
where ${\rm var}\{\cdot\}$ denotes variance. The matrices $\bf D$
and $\bf M$ reflect the line-of-sight (LOS) and scattering
components of the channel, respectively. The $(i,j)$-th element of
${\bf \Omega }$, i.e.\ $[{\bf \Omega }]_{ij}$, corresponds to the
average power of $[{\bf \tilde {H}}]_{ij}$ and captures the average
coupling between the $i$-th receive eigenmode and $j$-th transmit
eigenmode. For this reason, we refer to ${\bf \Omega }$ as the {\it
eigenmode channel coupling matrix}. It can be seen that the
eigenvalues of the transmit and receive correlation matrices are
summations of the elements of the matrix ${\bf \Omega }$ in each
column and each row, respectively. These eigenvalues are
\emph{non-separable}, which reflects the joint correlation feature
of the channel.

The channel model described by (\ref{channel_matrix}) provides a
general formula which embraces many existing channel models
\cite{Shui,Kermoal,Sayeed,Zhou,Weichselberger,Tulino2005,Tulino2006,Veeravalli}.
For example, if $\bf D =0$, $\bf M$ is a rank-one matrix, and ${\bf
H}_{iid}$ has Rayleigh-faded elements, then (\ref{channel_matrix})
reduces to the popular separable-correlation Kronecker model
\cite{Shui,Kermoal}. By allowing $\bf M$ to have arbitrary rank and
fixing ${\bf U}_{t}$ and ${\bf U}_{r} $ to be DFT matrices, one can
achieve the virtual channel representation for ULAs \cite{Sayeed}.
If we further allow ${\bf U}_{t}$ and ${\bf U}_{r}$ to be arbitrary
unitary matrices, one can obtain Weichselberger's channel
model\cite{Weichselberger}. Moreover, by setting $\bf D = 0$ we
arrive at the unitary-independent-unitary (UIU) model introduced in
\cite{Tulino2005,Tulino2006}. Our model is also related to the model
in \cite{Veeravalli}, where one LOS component was included in the
virtual channel representation for the ULA MIMO channels. Here, we
allow multiple LOS components in eigenmode to cover more general
transmission links, such as those in distributed radio networks
\cite{You05}.

\subsection{Statistical Eigenmode Transmission}
\label{sec: Statistic Eigenmode Transmission}

Throughout the paper, we assume that the receiver knows the channel
perfectly, whilst the transmitter only has access to the statistical
parameters, including ${\bf U}_{t}$, ${\bf U}_{r}$, ${\bf D}$ and
${\bf M}$ (and thus $\bf \Omega$). Under these assumptions, the
ergodic capacity of the MIMO channel is achieved by selecting the
transmitted signal vector $\bf x$ to have zero mean and to follow a
proper Gaussian distribution \cite{Telatar}. Let the covariance
matrix of $\bf x$ be $ \mathbb{E}\left\{ {{\mathbf{xx}}^H }
\right\}=\left( P/N_t \right) {\mathbf{Q}} $. Then the power
constraint on  $\bf x$ can be rewritten as ${\rm tr}\left(
{\mathbf{Q}}\right)=N_t$, and the ergodic capacity is given by
\begin{equation}\label{CAPACITY}
C = \mathop {\max }\limits_{{\bf Q}\succeq{\bf 0},\text{tr}\left(
{\mathbf{Q}} \right) = N_t} \mathbb{E}\left\{ {\log \det \left(
{{\mathbf{I}}_{N_r} + \gamma {{\mathbf{HQH}}^H  }} \right)}
\right\},
\end{equation}
where $\gamma=\frac{\rho}{N_t}$. Substituting (\ref{channel_matrix})
into (\ref{CAPACITY}) yields
\begin{equation}\label{CAPACITY_vitural}
C =\mathop {\max }\limits_{{\bf \tilde Q}\succeq{\bf
0},\text{tr}\left( {\mathbf{\tilde Q}} \right) = N_t}
\mathbb{E}\left\{ {\log \det \left( {{\mathbf{I}}_{N_r} + \gamma
{{\mathbf{{\tilde H} {\tilde Q}{\tilde H}}}^H }} \right)} \right\},
\end{equation}
where ${\tilde {\bf Q}}={\bf U}_{t}^{H} {\bf Q} {\bf U}_{t}$. Let
${\bf Q}={\bf U \Lambda U}^{H}$, where $\bf U$ is the eigenvector
matrix, and ${\bf
\Lambda}=\diag(\lambda_1,\lambda_2,\ldots,\lambda_{N_t})$ is a
diagonal matrix of the corresponding eigenvalues. When $\bf\tilde H$
has independent and symmetrically distributed elements, it has been
shown in \cite{Tulino2006} and \cite{Veeravalli} that the optimal
eigenvector matrix $\bf U$ for achieving the capacity is ${\bf
U}={\bf U}_{t}$, and thus ${\tilde {\bf Q}}$ is diagonal. In
\cite{Veeravalli}, it has been pointed out that this solution also
applies when one element of $\bf\tilde H$ contains a LOS component.
We note however, that the channel model given by
(\ref{channel_matrix}) allows for multiple possible LOS components.
For this more general case, one can arrive at the following result:
\begin{theorem}\label{optimal_trasmit_direction}
The eigenvector matrix of the capacity-achieving matrix $\bf Q$ for
the jointly-correlated channel (\ref{channel_matrix}) is given by
${\bf{U}} = {\bf{U}}_t$. The ergodic capacity can therefore be
expressed as
\begin{equation}\label{capacity_optimized}
C =\mathop {\max }\limits_{{\boldsymbol \lambda}\geq {\bf
0},\sum_{i=1}^{N_t} \lambda_i = N_t} \mathbb{E}\left\{ {\log \det
\left( {{\mathbf{I}}_{N_r} + \gamma {{\mathbf{\tilde
H}\,{\diag({\boldsymbol \lambda})}{\bf \tilde H}}}^{H} } \right)}
\right\},
\end{equation}
where $\boldsymbol \lambda$ is an $N_t\times 1$ vector containing
the eigenvalues $\lambda_i$, $i=1,2,\ldots , N_t$.
\end{theorem}

The proof follows similar approaches to those used in
\cite{Visotsky,Hosli,Tulino2006,Veeravalli} and is therefore
omitted. Theorem \ref{optimal_trasmit_direction} demonstrates that
the optimal signaling directions are defined by the eigenvectors of
the transmit correlation matrix of the MIMO channel.  This agrees
with and extends prior results in the literature to the more general
channel model given by (\ref{channel_matrix}). For the transmitter
optimization problem, the major remaining challenge is to determine
the eigenvalues of the capacity-achieving input covariance matrix
$\mathbf{Q}$. This is equivalent to the task of optimally allocating
the available transmit power budget over the optimized transmit
eigen-directions, determined in Theorem
\ref{optimal_trasmit_direction}.

In general, it is very difficult to derive an exact closed-form
solution for the power allocation problem.  A major source of this
difficulty is due to the complexity in evaluating tractable
closed-form solutions for the expectation in
(\ref{capacity_optimized}). This is also the case for many other
less general MIMO channel scenarios, such as the popular Kronecker
correlation model \cite{Jafar_comm,Jorswieck}.  As such, the
standard approach has been to develop numerical optimization
techniques (see e.g.\ \cite{Vu} and \cite{Hanlen}).

In this paper, considering the general jointly-correlated MIMO
channel model, we develop a new approach which leads to the design
of simple, robust and practical power allocation solutions. In
particular, our approach is based on deriving a tight closed-form
upper bound on the expectation in (\ref{capacity_optimized}) which
can then serve as an approximation to the capacity. Based on this
expression, we are then able to derive new optimized power
allocation solutions which are simple and fast to compute. These
power allocation solutions will be shown to serve as very accurate
approximations to the optimal capacity-achieving solution, with low
computational complexity requirements.

We note that the power-allocation problem for jointly-correlated
channel scenarios has also been considered in \cite{Tulino2006},
where necessary and sufficient conditions as well as an iterative
numerical algorithm were proposed. One drawback of that algorithm is
that for each iteration it requires numerically averaging certain
random matrix structures involving the inverse of instantaneous
realizations of the MIMO channel. Moreover, since the computation
algorithm requires access to instantaneous MIMO CSI, then under the
statistical-feedback assumption, such power-allocation computations
must be typically performed at the receiver. In contrast, in this
paper we develop more practically appealing power-allocation
algorithms which involve only the channel statistics. As such, they
are simpler and more efficient to compute, since they do not require
random matrix averaging during the power-allocation computation.
Moreover, our new power-allocation algorithm has the additional
advantage of permitting computation at either the receiver or the
transmitter. This extra flexibility is particularly important for
various practical applications, such as downlink transmission where
it is often desirable or necessary to restrict computations to the
base station.

We start by rewriting the ergodic capacity
(\ref{capacity_optimized}) as
\begin{equation}\label{capacity_mutual information}
C =\mathop {\max }\limits_{{\boldsymbol \lambda}\geq {\bf
0},\sum_{i=1}^{N_t} \lambda_i = N_t} {\Tilde C}(\boldsymbol
\lambda),
\end{equation}
where
\begin{equation}\label{mutual information}
{\Tilde C}(\boldsymbol \lambda)= \mathbb{E}\left\{ {\log \det \left(
{{\mathbf{I}}_{N_r} + \gamma {{\mathbf{\tilde
H}\,{\diag({\boldsymbol \lambda})}{\bf \tilde H}}}^{H} } \right)}
\right\}
\end{equation}
is the expected mutual information between the transmitted signal
$\bf x$ and the received signal $\bf y$ under SET. Due to the
concavity of the $\log (\cdot)$ function, the mutual information
${\Tilde C}(\boldsymbol \lambda)$ is upper bounded by
\begin{equation}\label{mutual information upbound}
{\Tilde C}(\boldsymbol \lambda)\leq {\Tilde C}_u(\boldsymbol
\lambda)=\log \mathbb{E}\left\{ \det \left( {{\mathbf{I}}_{N_r} +
\gamma {{\mathbf{\tilde H}\,{\diag({\boldsymbol \lambda})}{\bf
\tilde H}}}^{H} } \right) \right\}.
\end{equation}
Thus, the ergodic capacity is upper bounded by
\begin{equation}\label{capacity upbound}
C\leq C_u=\mathop {\max }\limits_{{\boldsymbol \lambda}\geq {\bf
0},\sum_{i=1}^{N_t} \lambda_i = N_t} {\Tilde C}_u(\boldsymbol
\lambda).
\end{equation}
For the case of Kronecker MIMO channels, it has been shown in
\cite{Shin,Kiessling_vtc,Kiessling_diss,Jin,Mckay} that such bounds
are very tight and admit closed-form expressions by using the
expansion of the determinant.

\section{Closed-Form Capacity Upper Bound using Permanents}\label{sec:Upper Bound}

In this section, we derive a closed-form expression for the capacity
upper bound (\ref{capacity upbound}) for the jointly-correlated MIMO
channel model in (\ref{channel_matrix}). We also develop algorithms
for its efficient computation. The upper bound derivation is based
heavily on exploiting linear-algebraic concepts and properties of
\emph{matrix permanents}, which we introduce and develop in the
sequel.

\subsection{Matrix Permanents: Definitions and Properties}

The permanent of a matrix is defined in a similar fashion to the
determinant. The primary difference is that when taking the
expansion over minors, all signs are positive
\cite{Minc,Ryser,Nijenhuis,Liang}. The permanents of square matrices
have been thoroughly investigated in linear algebra and various
applied fields. The permanents of $M\times N$ matrices with $M\leq
N$ have also been defined and investigated\cite{Minc}. In this
paper, to facilitate our capacity upper bound derivation we find it
necessary to extend the definition of permanents to allow for
arbitrary $M$ and $N$, and provide their useful properties.

\begin{definition} \label{def:Permanent}
For an $M \times N$ matrix $\bf A$, the permanent is defined as
\begin{equation}
\label{definition_permanents}
\Per {\left( \bf A \right)}=
\left\{\begin{array}{cc} \sum\limits_{{\hat{\alpha}_{M}}\in
\mathcal {S}_{N}^M}\prod\limits_{i=1}^{M} a_{i,{\alpha}_{i}}, &  M\leq N \\
\sum\limits_{{\hat{\beta}}_{N}\in \mathcal
{S}_{M}^N}\prod\limits_{i=1}^{N} a_{{\beta}_{i},i}, & M > N,
\end{array} \right.
\end{equation}
where $a_{i,j}$ denotes the $(i,j)$-th element of $\bf A$.
\end{definition}

From this definition, one can easily establish a number of important
properties of the matrix permanent, as given in the following lemma.
These properties will be useful in subsequent derivations.
\begin{lemma} \label{Per:Basic properties }
Let $\bf A $ be an $M \times N$ matrix, $\bf a$ an $M \times 1$
vector, $\bf b $ an $N \times 1$ vector, and $\mu$ a scale constant.
Then
\begin{align}\label{Per: property 1}
 \Per ({\bf A})&= \Per ({\bf A}^T) \\
\label{Per: property 2}
 \Per ({\bf a})&= \sum_{i=1}^{M}a_i\\
\label{Per: property 3}
 \Per ({\rm diag} (\bf a))&=\det ({\rm diag} (\bf a))\\
\label{Per: property 4}
 \Per (\mu {\bf A})&= \mu^{\min (M,N)}  \Per ({\bf
A})\\
\label{Per: property 5}
 \Per ({\rm diag} {(\bf a)}{\bf A })&=\det({\rm diag} { (\bf
a))} \Per ({\bf A}), \  M \leq N\\
\label{Per: property 6}
 \Per ({\bf A} {\rm diag} ({\bf b}) )&= \det({\rm diag} ({\bf
b})) \Per ({\bf A}),\  M \geq N.
\end{align}
\end{lemma}

For an $M\times N $ matrix with $M\leq N$, there exists an analogy
between the matrix permanent and the Laplace expansion of the
determinant\cite{Minc,Aitken}. The following lemma gives the straightforward
extension of this result for arbitrary\footnote{The result for the case $M > N$ is obtained by employing (\ref{Per: property 1}), and following the same steps as used in the derivation for the case $M\leq N$, given in \cite{Minc}.} $M$ and $N$.
\begin{lemma} \label{Lemma: Per:Laplace expansion}
Let $\bf A $ be an $M \times N$ matrix. Then
\begin{equation}\label{Per:Laplace expansion}
 \Per({\bf A})=\left\{\begin{array}{cc}
\sum\limits_{{\hat{\sigma}_{k}}\in \mathcal {S}_{N}^{(k)}}{ \Per
\left({\bf A}^{\hat{\alpha}_{k}}_{\hat{\sigma}_{k}}\right)}
{\Per \left({\bf A}^{\hat{\alpha}^{'}_{k}}_{\hat{\sigma}^{'}_{k}}\right)}, &  M \leq N \\
\sum\limits_{{\hat{\sigma }}_{k}\in \mathcal{S}_{M}^{(k)}}{ \Per
\left( {\bf A}^{\hat{\sigma}_{k}}_{\hat{\beta}_{k}}\right)}{ \Per
\left( {\bf A}^{\hat{\sigma}^{'}_{k}}_{\hat{\beta}^{'}_{k}}\right)},
& M > N,
\end{array} \right.
\end{equation}
where ${\hat{\alpha }_{k}}\in \mathcal {S}_{M}^{(k)}$ and
${\hat{\beta}_{k}}\in \mathcal {S}_{N}^{(k)}$ with $1 \leq k \leq
{\rm min}(M, N)$. Note that for the case $k={\rm min}(M, N)$, $\Per
\left({\bf A}^{\hat{\alpha}^{'}_{k}}_{\hat{\sigma}^{'}_{k}}\right)=1
$ and $\Per \left( {\bf
A}^{\hat{\sigma}^{'}_{k}}_{\hat{\beta}^{'}_{k}}\right)=1$.
\end{lemma}

For the special case $k=1$, (\ref{Per:Laplace expansion}) can be
re-expressed as follows
\begin{equation}\label{Per:Laplace expansion k=1}
 \Per({\bf A})=\left\{\begin{array}{cc}
\sum\limits_{j=1}^{N} a_{i, j}{ \Per \left({\bf A}_{i,j}\right)}
, &  M \leq N \\
\sum\limits_{j=1}^{M} a_{j, i}{ \Per \left({\bf A}_{j,i}\right)}, &
 M > N,
\end{array} \right.
\end{equation}
where $1 \leq i \leq {\rm min}(M, N)$. This is analogous to the
cofactor expansion of the determinant\cite{Aitken}. With $k={\rm
min}(M, N)$, (\ref{Per:Laplace expansion}) simplifies to
\begin{equation}\label{Per: property 7}
 \Per({\bf A})=\left\{\begin{array}{cc}
\sum\limits_{{\hat{\sigma}_{M}}\in
\mathcal {S}_{N}^{(M)}}{ \Per \left({\bf A}_{\hat{\sigma}_{N}}\right)}, &  M \leq N \\
\sum\limits_{{\hat{\sigma}}_{N}\in \mathcal{S}_{M}^{(N)}}{ \Per
\left( {\bf A}^{\hat{\sigma}_{N}}\right)}, &
 M > N.
\end{array} \right.
\end{equation}

The following two key lemmas are particularly important for deriving
the closed-form capacity upper bound in the sequel.
\begin{lemma} \label{lemma: Per:Polynomial expansion}
Let $\bf A $ be an $M \times
N$ matrix. Then
\begin{align}
\label{Per:EQ Polynomial expansion } \Per (\left[{\bf I}_{M} \, {\bf
A}\right])& = \Per ([{\bf
I}_{N} \, {\bf A}^{T} ]\nonumber \\
& = \sum \limits_{k=0}^{\min(N,M)} \sum\limits_{{\hat{\alpha}}_k\in
\mathcal{S}_M^{(k)}}\Per ({\bf A}^{\hat\alpha_k })\nonumber \\
& = \sum\limits_{k=0}^{\min(N,M)} \sum\limits_{{\hat{\beta}}_k\in
\mathcal{S}_N^{(k)}}\Per({\bf A}_{\hat\beta_k }),
\end{align}
where $\Per ({\bf A}^{\hat{\alpha}_k} )=1$ and $\Per ({\bf
A}_{{\hat{\beta}_k}} )=1$ when $k=0$.
\end{lemma}

A proof is provided in Appendix \ref{Appendix: Proof of lemma:
Per:Polynomial expansion}. The values of $\Per ([{\bf I}_{M} \, {\bf
A}])$ and $\Per ([{\bf I}_{N} \,{\bf A}^{T} ])$ in Lemma \ref{lemma:
Per:Polynomial expansion} will be called {\it extended permanents}
of $\bf A$, which we denote as
\begin{align} \label{eq:extendedPerms}
\underline{\Per}({\bf A})=\Per ([{\bf I}_{M} \, {\bf A}])=\Per
([{\bf I}_{N} \, {\bf A}^{T} ]) \; .
\end{align}

\begin{lemma} \label{Lemma: E_det_prod}
For an $N \times N$ random matrix ${\bf X}$ with independent
elements, suppose that there exists at most one non-zero element in
each row of $\bar{\bf X}=\mathbb{E}\left\{ \bf X \right\}$. Then we
have
\begin{equation}
\mathbb{E}\left\{ \det\left( \bf X \right) \det\left( {\bf X}^{H}
\right) \right\}=\Per \left({\bf \Xi} \right),
\end{equation}
where ${\bf \Xi}=\mathbb{E}\left\{ \bf{X \odot X}^{*}\right\}$.
\end{lemma}

A proof is provided in Appendix \ref{Appendix: Proof of Lemma:
E_det_prod}. For the special case where all elements of $\bf X$ are
independent and identically distributed with zero mean and unit
variance, we have that ${\bf \Xi}=\mathbb{E}\left\{ \bf{X \odot
X}^{*}\right\}={\bf 1}_{N \times N}$ and  $\mathbb{E}\left\{
\det\left( \bf X \right) \det\left( {\bf X}^{H} \right)
\right\}=\Per \left({\bf 1}_{N \times N} \right)=N!$. This agrees
with prior results in
\cite{Grant,Kiessling_vtc,Kiessling_diss,Shin}.

The following conjecture is useful when dealing with the optimal
power allocation problem in Section \ref{sec: Optimal Power
Allocation}.
\begin{conjecture} \label{Lemma:Concavity}
Let $\bf A$ be an $M\times N$ matrix of non-negative elements. Then
$f(\boldsymbol \lambda)=\log \Per({\bf A}\diag(\boldsymbol
\lambda))$ is concave on $\mathcal{D}^{N}=\{\boldsymbol \lambda \, |
\, \Per({\bf A}\diag(\boldsymbol \lambda))>0, \; {\rm and} \;
\lambda_{i} \geq 0, \, 1 \leq i \leq N \}$.
\end{conjecture}

For the general case with arbitrary $M$ and $N$, the formal proof of
this result is not available at this stage. In Appendix
\ref{Appendix: Proof of Lemma:Concavity}, we provide proofs for
several special cases, which lend support to the validity of this
conjecture.

\subsection{Capacity Upper Bound}
Armed with the general results of the preceding subsection, we can
now derive a closed-form expression for the upper bound on the
ergodic capacity.
\begin{theorem}\label{upper_bound}
The ergodic capacity in (\ref{capacity_optimized}) is upper bounded
by
\begin{equation}\label{upper bound_C_u(Lambda)}
C\leq C_{u}=\mathop {\max }\limits_{{\boldsymbol \lambda} \geq {\bf
0}, {\bf 1}^T {\boldsymbol \lambda} = N_t} \log {\Tilde
C}_{u}({\boldsymbol \lambda}),
\end{equation}
where
\begin{equation}\label{C_u(Lambda)}
{\Tilde C}_{u}({\boldsymbol \lambda})=\log
 \underline{\Per} \left( \gamma {\bf
\Omega} \, \diag ({\boldsymbol \lambda}) \right).
\end{equation}
\end{theorem}

\begin{proof}
We start by writing the upper bound for the expected mutual
information under SET in (\ref{mutual information upbound}) as
\begin{equation}\label{upcapacity}
{\tilde C}_{u}(\boldsymbol \lambda)= \log E(\boldsymbol \lambda)
\end{equation}
where
\begin{equation}\label{upcapacity_lambda}
E(\boldsymbol \lambda) = \mathbb{E} \left\{ {\det \left(
{\bf{I}}_{N_r} + \gamma  {\bf{\tilde H}}\,\diag( {\boldsymbol
\lambda }){\bf{\tilde H}}^{H} \right)} \right\} \; .
\end{equation}
By using the characteristic polynomial expansion of the determinant,
as well as the Cauchy-Binet formula for the determinant of a product
matrix\cite{Aitken}, we have
\begin{align}
\label{determinant expansion1} E({\boldsymbol \lambda })
&=\mathbb{E} \left\{\sum\limits_{k = 0}^{\min(N_t, N_r)} \gamma ^{k}
{\sum\limits_{{\hat \alpha _k}\in \mathcal {S}_{N_r}^{(k)} } {\det
\left(\left( {\bf{\tilde H}}\, \diag( {\boldsymbol \lambda
}){\bf{\tilde H}}^{H} \right)_{\hat
\alpha _k }^{\hat \alpha _k }\right) } }\right\} \nonumber \\
 &= \mathbb{E} \left\{ \sum\limits_{k = 0}^{\min(N_t, N_r)} \gamma
^{k} { \sum\limits_{\hat \alpha _k \in \mathcal {S}_{N_r}^{(k)}}\;
\sum\limits_{\hat \beta _k \in \mathcal {S}_{N_t}^{(k)}}\;
\sum\limits_{\hat \xi _k \in \mathcal {S}_{N_t}^{(k)} }{\det \left(
{\bf{\tilde H}} _{\hat \beta _k }^{\hat \alpha _k }\right) } } {
{\det \left( \diag( {\boldsymbol \lambda }) _{\hat \xi _k }^{\hat
\beta _k }\right) } } {{\det \left( \left({\bf{\tilde H}}^{H}\right)
_{\hat \alpha _k }^{\hat \xi _k
}\right) } }\right\} \nonumber \\
 &=  \sum\limits_{k = 0}^{\min(N_t, N_r)} \gamma ^{k}
\sum\limits_{\hat \alpha _k \in \mathcal {S}_{N_r}^{(k)}}\;
\sum\limits_{\hat \beta _k \in \mathcal {S}_{N_t}^{(k)}} {\det
\left( \diag( {\boldsymbol \lambda }) _{\hat \beta _k }^{\hat \beta
_k }\right) } \mathbb{E} \left\{ {\det \left(  {\bf{\tilde H}}_{\hat
\beta _k }^{\hat \alpha _k } \right) } {\det \left(
\bigl({\bf{\tilde H}}^{H}\bigr)_{\hat \alpha _k }^{\hat \beta_k }
\right) } \right\} \; .
\end{align}
Let us denote ${\bf X}=\left( {\bf{\tilde H}} \right)_{\hat \beta _k
}^{\hat \alpha _k }$. Then, ${\bf X}^{H}= \left({\bf{\tilde
H}}^{H}\right)_{\hat \alpha _k }^{\hat \beta_k }$, and it is easily
found that $\mathbb{E} \left\{ {\bf X}\odot {\bf X}^{*}\right\}={\bf
\Omega}_{\hat \beta _k }^{\hat \alpha_k }$. The matrix ${\bf X}$
satisfies the conditions in Lemma \ref{Lemma: E_det_prod}. Thus, we
have
\begin{align} \label{determinant expansion_E_det_prod}
\mathbb{E} \left\{ {\det \left( {\bf{\tilde H}}_{\hat \beta _k
}^{\hat \alpha _k } \right) }   {\det \left( \left({\bf{\tilde
H}}^{H} \right)_{\hat \alpha _k }^{\hat \beta_k } \right)}
\right\}&= \Per\left({\bf \Omega}_{\hat \beta _k }^{\hat \alpha_k
}\right).
\end{align}
Substituting (\ref{determinant expansion_E_det_prod}) into
(\ref{determinant expansion1}) and using the properties of the
permanents in Lemma \ref{Per:Basic properties }, as well as
(\ref{Per: property 7}) and Lemma \ref{lemma: Per:Polynomial
expansion}, we find that
\begin{align}
E({\boldsymbol \lambda }) &= \! \sum\limits_{k = 0}^{\min(N_t, N_r)}
\!\!\! \gamma ^{k} \!\!\! \sum\limits_{\hat \alpha _k  \in \mathcal
{S}_{N_r}^{(k)}}  \sum\limits_{\hat \beta _k \in \mathcal
{S}_{N_t}^{(k)}} \!\!\! {\det \left( \diag( {\boldsymbol \lambda })_{\hat
\beta _k }^{\hat \beta _k }  \right)} \Per \left({\bf \Omega}_{\hat
\beta _k }^{\hat \alpha_k }\right)\nonumber \\
&=\!\sum \limits_{k = 0}^{\min(N_t, N_r)} \!\!\! \gamma ^{k} \!\!\! \sum\limits_{\hat
\alpha _k \in \mathcal {S}_{N_r}^{(k)}} \sum\limits_{\hat \beta _k
\in \mathcal {S}_{N_t}^{(k)}} \!\!\! \Per\left( \left({\bf \Omega}\,\diag(
{\boldsymbol \lambda }) \right)_{\hat \beta _k }^{\hat \alpha_k }\right)\nonumber \\
&=\!\sum\limits_{k = 0}^{\min(N_t, N_r)} \!\!\! \gamma ^{k} \!\!\! \sum\limits_{\hat
\alpha _k \in \mathcal {S}_{N_r}^{(k)}} \!\!\! \Per\left(\left({\bf
\Omega}\,\diag( {\boldsymbol \lambda }) \right)^{\hat \alpha_k }\right)\nonumber \\
&=\!\sum\limits_{k = 0}^{\min(N_t, N_r)}  \sum\limits_{\hat \alpha _k
\in \mathcal {S}_{N_r}^{(k)}}  \!\!\! \Per \left(\left(\gamma{\bf
\Omega}\,\diag( {\boldsymbol \lambda }) \right)^{\hat \alpha_k
}\right) \nonumber \\
&= \underline{\Per} \left( \gamma {\bf \Omega}\diag( {\boldsymbol
\lambda }) \right).   \label{determinant expansion2}
\end{align}
Substituting (\ref{determinant expansion2}) into (\ref{upcapacity})
and using (\ref{capacity upbound}) complete the proof.
\end{proof}

From the above theorem, we see that the upper bound on capacity is
completely determined by the average SNR $\rho$ ($=\gamma N_t$) and
the eigenmode channel coupling matrix ${\bf \Omega}$. This bound is
particularly useful, since we may now apply (\ref{C_u(Lambda)}) to
maximize $\tilde{C}_{u}({\boldsymbol \lambda})$ with the respect to
$\boldsymbol \lambda$ (i.e.\ address the power allocation problem),
without the need for performing Monte-Carlo averaging over random
realizations of the MIMO channel matrix.

It is interesting to consider the special case when $\bf D = 0$ and
${\bf M= a b}^{T}$, where $\bf a$ and $\bf b$ are $N_r\times1$ and
$N_t\times 1$ real vectors. In this case, the jointly-correlated
channel model considered in this paper reduces to the popular
Kronecker correlation model. Defining ${\boldsymbol \lambda}_r=\bf
a\odot a$ and ${\boldsymbol \lambda}_t=\bf b\odot b$, the eigenmode
channel coupling matrix can then be expressed as ${\bf
\Omega}={\boldsymbol \lambda}_{r} {\boldsymbol \lambda}_{t}^{T} =
{\diag}\left({\boldsymbol \lambda}_{r}\right) \, {\bf 1}_{N_{r}
\times N_{t}}\,{\diag}\left({\boldsymbol \lambda}_{t}\right) $, and
(\ref{C_u(Lambda)}) reduces to
\begin{align}
{\Tilde C}_{u}({\boldsymbol \lambda }) & = \log \underline{\Per}
\left( \gamma {\bf
\Omega} \,\diag( {\boldsymbol \lambda }) \right) \nonumber \\
& =\log \sum\limits_{k = 0}^{\min(N_t, N_r)} \gamma ^{k}
\sum\limits_{\hat \alpha _k  \in \mathcal {S}_{N_r}^{(k)}}
\;\sum\limits_{\hat \beta _k \in \mathcal {S}_{N_t}^{(k)}} {\det
\left( \diag( {\boldsymbol \lambda })_{\hat \beta _k }^{\hat \beta
_k } \right)} \Per \left(\left({\diag}\left({\boldsymbol
\lambda}_{r}\right) \, {\bf 1}_{N_{r} \times
N_{t}}\,{\diag}\left({\boldsymbol \lambda}_{t}\right) \right)_{\hat
\beta _k }^{\hat
\alpha_k }\right)\nonumber \\
& =\log \sum\limits_{k = 0}^{\min(N_t, N_r)} \gamma ^{k} \, k!
\sum\limits_{\hat \alpha _k  \in \mathcal {S}_{N_r}^{(k)}}
\;\sum\limits_{\hat \beta _k \in \mathcal {S}_{N_t}^{(k)}}{\det
\left( \diag( {\boldsymbol \lambda}_{r})_{\hat \alpha _k }^{\hat
\alpha_k } \right)}\, {\det \left( \diag( {\boldsymbol \lambda
}\odot {\boldsymbol \lambda}_{t})_{\hat \beta _k }^{\hat \beta _k }
\right)}. \label{eq:KroneckerCap}
\end{align}
Equation (\ref{eq:KroneckerCap}) is equivalent to the upper bounds
presented previously for Kronecker-correlated channels in
\cite{Kiessling_vtc,Kiessling_diss}. Moreover, for the special case
${\boldsymbol \lambda}={\bf 1}$ (i.e.\ the case of equal-power
allocation), (\ref{eq:KroneckerCap}) reduces further to the capacity
upper bound presented in \cite{Shin}.

\subsection{Efficient Computation Algorithms}

To evaluate the closed-form capacity upper bound expression given by
({\ref{upper bound_C_u(Lambda)}}) and (\ref{C_u(Lambda)}), we must
evaluate the extended permanent of the matrix $\gamma {\bf
\Omega}\,\diag({\boldsymbol \lambda})$. Clearly, when the size of
the matrix is small, this can be done by simply expressing the
extended permanent as a conventional permanent via
(\ref{eq:extendedPerms}), and then either directly applying
Definition \ref{def:Permanent}, or using the Laplace expansion in
Lemma \ref{Lemma: Per:Laplace expansion}. However, in both cases, as
the size of the matrix grows, the computational complexity increases
significantly, and more efficient methods are needed. To see this,
consider the task of evaluating the permanent of a general $M\times
N$ matrix $\bf A$. The complexity associated with computing matrix
permanents is conventionally measured in terms of the number of the
required multiplications. Adopting this measure, the number of
multiplications required for evaluating the matrix permanent using
Definition \ref{def:Permanent} and the Laplace expansion (e.g. via
(\ref{Per:Laplace expansion k=1})) are $\frac{(m-1)n!}{(n-m)!}$ and
$\sum_{k=1}^{m-1}\frac{n!}{(n-k)!}$, respectively, where $m={\rm
min}(M,N)$ and $n={\rm max}(M,N)$. Clearly, as the matrix dimensions
increase, the computational complexity increases exponentially. For
this reason, it is necessary to investigate more efficient
computational algorithms.

The best-known algorithm for computing the matrix permanent of
arbitrary dimensions is due to Ryser \cite{Ryser}\footnote{For the
case of square matrices, further improvements have also been
proposed \cite{Nijenhuis}.}, who showed that the permanent of the
$M\times N$ matrix $\bf A$ (with $M\leq N$) can be evaluated via
\begin{equation}\label{Per:Ryser}
\Per({\bf A})=\sum \limits_{k=0}^{M} (-1)^{M-k} C_{N-k}^{M-k}\sum
\limits_{\hat{\alpha}_k \in
\mathcal{S}_{N}^{(k)}}\prod\limits_{i=1}^{M}r_i{({\bf
A}_{\alpha_k})},
\end{equation}
where $C_j^{i}=\frac{j!}{i!(j-i)!}$, and $r_i(\cdot)$ represents the
sum of the elements in the $i$-th row of the matrix argument. A
similar formula also exists for the case $M>N$. This algorithm
requires $m+(m-1)\sum_{k=1}^{m}C_{n}^{k}$ multiplications, with $m$
and $n$ defined as above.

In our case, we are interested in computing the extended permanent
$\underline{\Per} ( \hat{\bf \Omega} )$ in (\ref{C_u(Lambda)}),
i.e.\ the permanent of $[{\bf I}_{N_r} \, \hat{\bf \Omega }]$ or
$[{\bf I}_{N_t} \, \hat{\bf \Omega }^T]$, where $\hat{\bf
\Omega}=\gamma {\bf \Omega}\,\diag({\boldsymbol \lambda})$. By
directly computing this quantity based on Definition
\ref{def:Permanent}, the Laplace expansion, or Ryser's method, the
number of required multiplications is
$\frac{(N_{\rm min}-1)(N_{\rm min}+N_{\rm max})!}{N_{\rm max}!}$,
$\sum_{k=1}^{N_{\rm min}-1}\frac{(N_{\rm min}+N_{\rm max})!}{(N_{\rm min}+N_{\rm max}-k)!}$
and
$N_{\rm min}+(N_{\rm min}-1)\sum_{k=1}^{N_{\rm min}}C_{N_{\rm min}+N_{\rm max}}^{k}$,
respectively, where $N_{\rm min}={\rm min}(N_t,$ $ N_r)$ and
$N_{\rm max}={\rm max}(N_t,N_r)$. For practical values of $N_r$ and
$N_t$, these complexities can be quite high. As such, we are
motivated to establish new and more efficient methods for computing
the extended permanent, which we now consider.

Let us define the following auxiliary matrix
\begin{equation}
\hat{\bf \Omega }(z)={\bf 1}_{N_r \times N_t}+z\hat{\bf \Omega}.
\end{equation}
We will establish new efficient computation algorithms for
$\underline{\Per} ( \hat{\bf \Omega} )$ based on the following key
result.
\begin{lemma} \label{Lemma: Per: poly in z}
Let $\Per ({\hat{\bf \Omega }(z)})=\sum _{k=0}^{{\rm min}(N_r,N_t)
}\mu_k z^{k}$. Then
\begin{equation}\label{Per:poly_coeff}
\underline{\Per}( \hat{\bf \Omega })=\sum _{k=0}^{{\rm min}(N_r,
N_t)}\mu_k c_k ,
\end{equation}
where $c_k=|N_t-N_r|!/({\rm max}(N_r, N_t)-k)!$.
\end{lemma}

A proof is presented in Appendix \ref{Appendix: Proof of Lemma: Per:
poly in z}. This result shows that the extended permanent
$\underline{\Per}( \hat{\bf \Omega })$ can be calculated directly
from the polynomial expansion of  $\Per ({\hat{\bf \Omega }(z)})$.
Considering the case $N_r \leq N_t$, from Definition
\ref{def:Permanent} in (\ref{definition_permanents}), Laplace
expansion (\ref{Per:Laplace expansion k=1}) and Ryser's expression
(\ref{Per:Ryser}), we have the following three formulas for
$\Per({\hat{\bf \Omega }(z)})$:
\begin{equation}
\label{Extended permanents with definition} \Per({\hat{\bf \Omega
}(z)})= \sum\limits_{{\hat{\alpha}_{N_r}}\in \mathcal
{S}_{N_t}^{N_r}}\prod\limits_{i=1}^{N_r} \left(1+{\hat
\omega}_{i,{\alpha}_{i}}z\right),
\end{equation}
\begin{equation}\label{Extended permanents with Laplace expansion k=1}
\Per({\hat{\bf \Omega }(z)})=\sum\limits_{j=1}^{N_t} (1+{\hat
\omega}_{i,j}z){ \Per \left({\hat{\bf \Omega }(z)}_{i,j}\right)},
\end{equation}
\begin{equation}\label{Per:Ryser Poly}
\Per({\hat{\bf \Omega }(z)})=\sum \limits_{k=0}^{N_r} (-1)^{N_r-k}
C_{N_t-k}^{N_r-k}\sum \limits_{\hat{\alpha}_k \in
\mathcal{S}_{N_t}^{(k)}}\prod\limits_{i=1}^{N_r}r_i{({\hat{\bf
\Omega }(z)}_{\alpha_k})}.
\end{equation}
It is convenient to re-express (\ref{Per:Ryser Poly}) by letting
\begin{equation}\label{Prod poly}
\prod\limits_{i=1}^{N_r}(1+r_i{({\hat{\bf \Omega
}}_{\alpha_k})z})=\sum \limits_{i=0}^{N_r}a_{i,\alpha_k}z^{i},
\end{equation}
with $a_{0,\alpha_k}=1$, such that
\begin{equation}\label{Ryser PolyProd}
\prod \limits_{i=1}^{N_r}r_i{({\hat{\bf \Omega
}(z)}_{\alpha_k})}=\sum _{i=0}^{N_r}k^{N_r-i}a_{i,\alpha_k}z^{i} \;
.
\end{equation}
This yields
\begin{equation}\label{Per:Ryser Poly2}
\Per({\hat{\bf \Omega }(z)})=\sum \limits_{k=0}^{N_r}\sum
_{i=0}^{N_r} z^{i}(-1)^{N_r-k} C_{N_r-k}^{N_t-k}k^{N_r-i}\sum
\limits_{\hat{\alpha}_k \in \mathcal{S}_{N_t}^{(k)}}a_{i,\alpha_k}.
\end{equation}
Importantly, we find that each of the equivalent expressions
(\ref{Extended permanents with definition}), (\ref{Extended
permanents with Laplace expansion k=1}) and (\ref{Per:Ryser Poly2})
admit simple and efficient recursive algorithms for calculating the
coefficients of $z$. To demonstrate this, consider (\ref{Extended
permanents with definition}). Let $ {\tilde b}_k(z)=\prod_{i=1}^{k}
\left(1+{\hat \omega}_{i,{\alpha}_{i}}z\right)=1+\sum_{n=1}^k
b_{k,n}z^n$, where $k=1,2,\cdots N_r$. Then, ${\tilde
b}_{k+1}(z)={\tilde b}_k(z)\left(1+{\hat
\omega}_{k+1,{\alpha}_{k+1}}z\right)$ for $1\leq k \leq N_r-1$, and
therefore the coefficients of $z$ can be evaluated recursively via
\begin{equation}\label{Recursive formula for product term}
b_{k+1,n}=\left\{\begin{array}{cc}{\hat
\omega}_{k+1,\alpha_{k+1}}+b_{k,1},& n=1\\{\hat
\omega}_{k+1,\alpha_{k+1}}b_{k,n-1}+b_{k,n},& 2\leq n\leq k\\{\hat
\omega}_{k+1,\alpha_{k+1}}b_{k,k},& n=k+1.\end{array}\right.
\end{equation}
This result, combined with Lemma \ref{Lemma: Per: poly in z},
presents an efficient algorithm for computing the extended permanent
$\underline{\Per}( \hat{\bf \Omega })$. In a similar manner,
efficient computational algorithms can also be easily obtained based
on (\ref{Extended permanents with Laplace expansion k=1}) and
(\ref{Per:Ryser Poly2}). We omit the specific details of these. For
arbitrary $N_t$ and $N_r$, with $N_{\rm min}$ and $N_{\rm max}$ defined as
above, the number of required multiplications for the three
polynomial-based computation algorithms are
$\frac{N_{\rm min}(N_{\rm min}-1)N_{\rm max}!}{2(N_{\rm max}-N_{\rm min})!}$,
$\sum_{k=1}^{N_{\rm min}-1}\frac{(N_{\rm min}-k)N_{\rm max}!}{(N_{\rm max}-k)!}$ and
$N_{\rm min}^2+\frac{N_{\rm min}(N_{\rm min}-1)}{2}\sum_{k=1}^{N_{\rm min}}C_{N_{\rm max}}^{k}$,
respectively.

Fig. \ref{fig:fig0} presents the number of required multiplications
for evaluating ${\Tilde C}_u(\boldsymbol \lambda)$ based on the
three polynomial-based computation algorithms, for various antenna
configurations of the form $N=N_{t}=N_r$. The number of required
multiplications for calculating ${\Tilde C}_u(\boldsymbol \lambda)$
by \emph{directly} using Definition \ref{def:Permanent}, Laplace
expansion and Ryser's formula are also shown for comparison. We
clearly see that the polynomial-based algorithms have significantly
reduced computational complexity compared with the direct methods;
in many cases yielding orders of magnitude improvements. Of the polynomial-based algorithms,
the Laplace expansion gives the least complexity for $N \leq 5$,
whereas the Ryser-based formula is most efficient for $N>5$.

\section{Optimal Power Allocation with the Capacity Bound}
 \label{sec: Optimal Power Allocation}
\subsection{Asymptotic Optimality at Low and High SNR}
\label{sec:OPA_Asymptotic} Based on the tight closed-form capacity
upper bound in Theorem \ref{upper_bound}, we can now address the
transmitter power allocation optimization problem by dealing with
only the eigenmode channel coupling matrix $\bf \Omega $ and the
transmit SNR $\rho$ ($=\gamma N_t$). The optimal solution for
maximizing the upper bound will then serve as an approximation to
the optimal capacity-achieving power allocation solution.  Our
numerical results will confirm the accuracy of this approximation.

The power allocation optimization problem can be formulated as follows
\begin{eqnarray}\label{Optimization: formula objective}
\max_{\boldsymbol \lambda} && {\tilde C}_{u}({\boldsymbol \lambda}) \;\;\;\;\;\;\;\;\;\\
\label{Optimization: formula constraints} \mbox{ subject to} &&
{\boldsymbol \lambda} \geq {\bf 0}, {\bf 1}^T {\boldsymbol \lambda}
= N_t.
\end{eqnarray}
Before dealing with this problem in its most generality, we briefly
check the asymptotic optimality of our approach at low and high SNR.
For arbitrary SNRs, we will then develop optimality conditions and
an iterative numerical algorithm in the framework of convex
optimization.

For low SNRs, ${\Tilde C}_{u}({\boldsymbol \lambda})$ can be expressed as
\begin{align}\label{Optimization: Low SNR}
{\Tilde C}_{u}({\boldsymbol \lambda})&=\log \left(1+ \gamma \sum
\limits_{i=1}^{N_t} \tau_i \lambda_i+O(\gamma^2)\right)\nonumber\\
&= \gamma \sum \limits_{i=1}^{N_t} \tau_i \lambda_i+O(\gamma^2),
\end{align}
where $\tau_i=\sum_{j=1}^{N_r} [{\bf \Omega}]_{ji}$. Without any
loss of generality, assume that
$\tau_1=\tau_2=\ldots=\tau_l>\tau_{l+1}\geq \ldots\geq\tau_{N_t}$.
Maximizing the first-order (in $\gamma$) term  in
(\ref{Optimization: Low SNR}) subject to the constraint
(\ref{Optimization: formula constraints}) gives the following
power-allocation policy
\begin{align} \label{eq:beamforming}
\lambda_i =
\left\{ \begin{array}{ll}
N_t/l , & {\rm for} \; i = 1, \ldots, l \\
0 , & {\rm for} \; i = l+1, \ldots, N_t.
\end{array}
\right.
\end{align}
This means that beamforming along the strongest transmit eigenmodes
(specified by the channel coupling matrix $\mathbf{\Omega}$) is
optimal in the low SNR case.

For high SNRs, with $N_t\leq N_r$, 
we have
\begin{align}\label{Optimization: High SNR}
{\Tilde C}_{u}({\boldsymbol \lambda})&\rightarrow\log \left(\Per
(\gamma \Omega \, \diag(\lambda) ) \right)\nonumber\\
&= \log \Per (\gamma \Omega )  + \log
\det\left(\diag(\lambda)\right)
\end{align}
which is maximized by the following power allocation policy
\begin{equation} \label{eq:EqualPower}
\lambda_i = 1, \; \;\;i = 1, \ldots, N_t
\end{equation}
i.e.\ equal-power allocation over the transmit eigenmodes. These low and high SNR power allocation policies, derived based on the capacity upper bound, coincide exactly with the optimal capacity-achieving power allocation policies for the low and high SNR regimes, considered previously in \cite{Tulino2006,Veeravalli}.

\subsection{Optimality Conditions for Arbitrary SNRs}
We now address the general case with arbitrary SNRs. To this end,
let ${\boldsymbol \lambda}_{1}\geq 0$, ${\boldsymbol
\lambda}_{2}\geq 0$, and $0 \leq \theta \leq 1$.  Then, using
Conjecture \ref{Lemma:Concavity}, we can write
\begin{equation}\label{eq:concave_proof}
\begin{split}
C_{u}(\theta {\boldsymbol \lambda}_1 +(1-\theta){\boldsymbol
\lambda} _2) &= \log \Per \Bigl ( \theta [{\bf I}_{N_r}\,\gamma
{\bf \Omega}\, {\diag ({\boldsymbol \lambda}_1)} ]+(1-\theta) [{\bf
I}_{N_r}\,\gamma {\bf \Omega} \, \diag ({\boldsymbol \lambda}_2) ]
\Bigr ) \\
&= \log \Per \Bigl ( \theta [{\bf
I}_{N_r}\,\gamma {\bf \Omega} ]{\diag (\tilde{{\boldsymbol
\lambda}}_1)}+(1-\theta)
[{\bf I}_{N_r}\,\gamma {\bf \Omega} ]\diag (\tilde{{\boldsymbol \lambda}}_2)  \Bigr )\\
& \geq  \theta \log \Per\left( [{\bf I}_{N_r}\,\gamma {\bf \Omega}
 ]\diag (\tilde{{\boldsymbol \lambda}}_1)\right) + (1-\theta)\log \Per \left(
[{\bf I}_{N_r}\,\gamma {\bf \Omega}   ]\diag (\tilde{{\boldsymbol \lambda}}_2)\right) \\
& =  \theta {\Tilde C}_{u}({\boldsymbol \lambda}_1 ) +(1-\theta)
{\Tilde C}_{u}({\boldsymbol \lambda}_2),
\end{split}
\end{equation}
where $\tilde{{\boldsymbol \lambda}}_1=[{\bf 1}_{1\times
N_r}\;{\boldsymbol \lambda}_1^T]^T$ and $\tilde{{\boldsymbol
\lambda}}_2=[{\bf 1}_{1\times N_r}\;{\boldsymbol \lambda}_2^T]^T$.
Therefore, the function ${\Tilde C}_{u}({\boldsymbol \lambda})$ is
concave on the space of nonnegative ${\boldsymbol \lambda}$, and the
optimization problem given by (\ref{Optimization: formula
objective}) and (\ref{Optimization: formula constraints}) is a
concave optimization problem. As such, there exists only one local
optimal solution, which is also a global solution.  This solution
could be evaluated by  employing standard convex optimization
algorithms, such as interior point methods \cite{Boyd}.

Since the problem is concave, we can derive necessary and sufficient conditions for the optimal solution using the
Karush-Kuhn-Tucker (KKT) conditions. To this end, let ${\boldsymbol \mu}=[\mu_1, \mu_2, ..., \mu_{N_t}]^T$
and $\nu$ be the Lagrange multipliers for the inequality constraint
${\boldsymbol \lambda} \geq 0$ and the equality constraint ${\bf
1}^T{\boldsymbol \lambda} = N_t $ respectively. Then the KKT
conditions satisfied by the optimal ${\boldsymbol \lambda}$ can be
expressed as
\begin{equation}\label{Optimization: KKT conditions }
\frac{\partial {\Tilde C}_{u}({\boldsymbol \lambda})}{\partial
\lambda_i}+ \mu_{i}+\nu=0,
\end{equation}
\begin{equation}\label{Optimization: KKT conditions others}
{\boldsymbol \lambda} \geq 0, \; \; {\bf 1}^T{\boldsymbol \lambda}=
N_t, \; \; {\boldsymbol \mu} \geq 0, \; \; \mu_i\lambda_i=0,
\end{equation}
where $\frac{\partial {\Tilde C}_{u}({\boldsymbol
\lambda})}{\partial \lambda_i}$ denotes the partial derivative of
${\Tilde C}_{u}({\boldsymbol \lambda})$ with respect to $\lambda_i$,
for $1 \leq i \leq N_t$.  From (\ref{C_u(Lambda)}), these
derivatives can be written as
\begin{equation}\label{Optimization:Derivatives}
\frac{\partial {\Tilde C}_{u}({\boldsymbol \lambda})}{\partial
\lambda_i}= \frac{1}{E({\boldsymbol \lambda})}\frac{\partial
E({\boldsymbol \lambda})}{\partial \lambda_i} ,
\end{equation}
where $E({\boldsymbol \lambda}) = \underline{\Per} \left( \gamma
{\bf \Omega} \, \diag ({\boldsymbol \lambda}) \right)$. To evaluate
the remaining derivatives in (\ref{Optimization:Derivatives}) it is
useful to apply the Laplace expansion property of permanents, given
by (\ref{Per:Laplace expansion k=1}), to express $E({\boldsymbol
\lambda})$ as follows
\begin{equation}\label{Optimization: Extended Per Laplace}
E({\boldsymbol \lambda})= p({\boldsymbol \lambda}_{(i)})+ \lambda_i
q({\boldsymbol \lambda}_{(i)}),
\end{equation}
where
\begin{align}\label{Optimization: Extended Per Laplace_p}
p({\boldsymbol \lambda}_{(i)})&= \underline{\Per}\left(\gamma {\bf
\Omega}_{(i)}\diag({\boldsymbol \lambda}_{(i)})\right),
\\
\label{Optimization: Extended Per Laplace_q} q({\boldsymbol
\lambda}_{(i)}) & = \sum\limits_{j=1}^{N_t}\gamma \omega_{j,i}
\underline{\Per}\left(\gamma {\bf \Omega}_{(i)}^{(j)}
\diag({\boldsymbol \lambda}_{(i)})\right) \nonumber \\
& = \underline{\Per}\left(\gamma {\bf \Omega}\, \diag({\boldsymbol
\lambda}_{i}) \right)-\underline{\Per}\left(\gamma {\bf
\Omega}_{(i)}\, \diag({\boldsymbol \lambda}_{(i)}) \right),
\end{align}
$\omega_{j,i}$ denotes the $(j,i)$-th element of $\bf \Omega$, ${\bf
\Omega}_{(i)}$ denotes the sub-matrix of $\bf \Omega$ obtained by
deleting the $i$-th column, ${\bf \Omega}_{(i)}^{(j)}$ denotes the
sub-matrix of $\bf \Omega$ obtained by deleting the $j$-th row and
$i$-th column, ${\boldsymbol \lambda}_{(i)}$ denotes the
$(N_t-1)\times 1$ vector obtained by deleting the $i$-th element of
${\boldsymbol \lambda}$, and ${\boldsymbol \lambda}_i$ denotes the
$N_t\times 1$ vector obtained by replacing the $i$-th element of
${\boldsymbol \lambda}$ by unity. Therefore,
(\ref{Optimization:Derivatives}) becomes
\begin{equation}\label{Optimization: Partial Derivative}
\frac{\partial {\Tilde C}_{u}({\boldsymbol \lambda})}{\partial
\lambda_i} = \frac{q({\boldsymbol \lambda}_{(i)})}{p({\boldsymbol
\lambda}_{(i)})+ \lambda_i q({\boldsymbol \lambda}_{(i)})} \; .
\end{equation}
Substituting (\ref{Optimization: Partial Derivative}) into
(\ref{Optimization: KKT conditions }) and eliminating the slack
variable ${\boldsymbol \mu}$, the KKT conditions become
\begin{equation}\label{Optimization: KKT conditions Simplified}
\lambda_i=\left(\tilde{\nu}-\frac{p({\boldsymbol \lambda}_{(i)})}{
q({\boldsymbol \lambda}_{(i)})}\right)^{+},
\end{equation}
\begin{equation} \label{Optimization: KKT conditions Power Constraint}
 {\bf
1}^T{\boldsymbol \lambda}= N_t,
\end{equation}
where $(a)^{+}=\max\{0, a\}$ and $\tilde{\nu}=1/\nu$.

In summary, we have the following theorem.
\begin{theorem} \label{Optimization: Concave and Conditions Theorem}
The expected mutual information upper bound ${\Tilde
C}_{u}({\boldsymbol \lambda})$ is concave with respect to
${\boldsymbol \lambda}$, and the necessary and sufficient conditions
for optimal power allocation are given by (\ref{Optimization: KKT
conditions Simplified}), where $\tilde{\nu}$ is chosen to satisfy
the power constraint in (\ref{Optimization: KKT conditions Power
Constraint}).
\end{theorem}

Note that when the eigenmode channel coupling matrix
$\mathbf{\Omega}$ is square and diagonal\footnote{In this special
case, the MIMO channel is essentially reduced to a set of
non-interfering scalar subchannels.}, we have
\begin{equation}
q({\boldsymbol \lambda}_{(i)})=\omega_{i,i} p({\boldsymbol
\lambda}_{(i)}) \; ,
\end{equation}
and the conditions in (\ref{Optimization: KKT conditions
Simplified}) simplify to
\begin{equation}\label{Optimization: KKT conditions Square and Diagonal}
\lambda_i=\left(\tilde{\nu}-\frac{1}{\omega_{i,i}}\right)^{+}.
\end{equation}
This is the same formula as the water-filling solution when the
transmitter has instantaneous CSI \cite{Telatar}, and one can easily
obtain the optimal power allocation via the water-filling algorithm.
However, in the general case of an arbitrary eigenmode channel
coupling matrix, the solution can not be obtained as easily and
numerical approaches are required.

\subsection{Iterative Water-Filling Algorithm}
In this section, we propose a simple iterative water-filling
algorithm (IWFA) for evaluating the optimal power allocation policy
which satisfies (\ref{Optimization: KKT conditions Simplified}). Our
algorithm is based on observing that the right-hand side of
(\ref{Optimization: KKT conditions Simplified}) is independent of
$\lambda_{i}$, and is motivated by the IWFA methods proposed in
\cite{Yu,Jindal} for transmitter optimization of multiuser systems
with instantaneous CSI known to the transmitters. Simulation
results, given in Section \ref{sec: Simulations}, show that this
approach works very well and is highly efficient; typically
converging after only a few iterations, with the first iteration
achieving near-optimal performance.  The proposed algorithm includes
the following steps:

\begin{description}
 \item [(1)]Initialize ${\boldsymbol \lambda}^{0}={\bf 1}$, ${\Tilde C}_{u}({\boldsymbol \lambda}^{0})
 =\log\underline{\Per}(\gamma {\bf \Omega})$, and
 $k=0$.
 \item [(2)]Calculate $p({\boldsymbol
\lambda}_{(i)}^{k})= \underline{\Per}(\gamma {\bf \Omega}_{(i)}\,
\diag({\boldsymbol \lambda}_{(i)}^{k}))$ and $q({\boldsymbol
\lambda}_{(i)}^{k})=\underline{\Per}\left(\gamma {\bf \Omega} \,
\diag({\boldsymbol \lambda}_{i}^{k}) \right)-\underline{\Per}(\gamma
{\bf \Omega}_{(i)}$ \sloppy $\diag({\boldsymbol \lambda}_{(i)}^{k})
)$, $i=1,2, ..., N_t$.
\item [(3)] Calculate $\lambda_i^{k+1}=(\tilde{\nu}-\frac{p({\boldsymbol
\lambda}_{(i)}^{k})}{q({\boldsymbol \lambda}_{(i)}^{k})})^+$,
$i=1,2, ..., N_t$, via the conventional water-filling algorithm with
power constraint $\sum_{i=1}^{N_t}\lambda_i^{k+1}=N_t$.
\item [(4)] Calculate ${\Tilde C}_{u}({\boldsymbol \lambda}^{k+1})=
\log\underline{\Per}(\gamma {\bf \Omega}\,\diag{({\boldsymbol \lambda}^{k+1})})$
.
\item [(5)] If ${\Tilde C}_{u}({\boldsymbol \lambda}^{k+1})\leq {\Tilde C}_{u}({\boldsymbol
\lambda}^{k})$,  set ${\boldsymbol
\lambda}^{k+1}:=\frac{1}{N_{t}}{\boldsymbol
\lambda}^{k+1}+\frac{N_{t}-1}{N_{t}}{\boldsymbol \lambda}^{k}$, and
recalculate ${\Tilde C}_{u}({\boldsymbol \lambda}^{k+1})$.
\item [(6)] Set $k:=k+1$ and return to Step 2 until the algorithm
converges or the iteration number is equal to a predefined value.
\end{description}
Here, ${\boldsymbol \lambda}^k$ stands for the value of
${\boldsymbol \lambda}$ in the $k$-th iteration. In Step 1 in the
first iteration, ${\boldsymbol \lambda}$ is initialized to $\bf 1$,
i.e., to the equal power allocation. Note, however, that
${\boldsymbol \lambda}$ could also be initialized in a different
way. For example, in practice it is reasonable to suppose that the
channel statistics change smoothly from frame to frame, where a more
appropriate starting point for any given frame would be the optimal
value of ${\boldsymbol \lambda}$ from the previous frame. This could
speed up the convergence of the IWFA. In Step 3, the conventional
water-filling algorithm is performed with the required variables
$p({\boldsymbol \lambda}_{(i)})$ and $q({\boldsymbol
\lambda}_{(i)})$ calculated in Step 2. Following the calculation of
the ${\Tilde C}_u(\boldsymbol \lambda)$ in Step 4, Step 5 is
performed to guarantee the convergence of the iterative procedure.
We discuss this issue in detail below. In Step 6, the convergence of
the algorithm can be determined by checking whether $| {\Tilde
C}_{u}({\boldsymbol \lambda}^{k+1})-{\Tilde C}_{u}({\boldsymbol
\lambda}^{k} )|$ (or $\|{\boldsymbol \lambda}^{k+1}-{\boldsymbol
\lambda}^{k} \|$) is less than some predefined tolerance.

\begin{theorem}\label{Optimization: IWFA convergence}
The IWFA for optimal power allocation converges to the capacity
upper bound $C_u$.
\end{theorem}

\begin{proof}
In order to verify the convergence of our proposed IWFA for
optimal power allocation, we define the following function for a
given ${\boldsymbol \lambda}^{k}$:
\begin{equation}\label{Optimization: Auxiliary function}
\bar{C}_u({\boldsymbol \lambda})= \frac{1}{N_t} \sum
\limits_{i=1}^{N_t} \log \left(p({\boldsymbol \lambda}_{(i)}^{k})+
\lambda_i q({\boldsymbol \lambda}_{(i)}^{k})\right).
\end{equation}
It can be seen that $\bar{C}_u({\boldsymbol \lambda})$ is a concave
function with respect to ${\boldsymbol \lambda}$. The water-filling
solution in Step 3 of the IWFA is exactly equal to the solution of
maximizing $\bar{C}_u({\boldsymbol \lambda})$, for a given
${\boldsymbol \lambda}^{k}$ subject to the power constraint ${\bf
1}^T {\boldsymbol \lambda}=N_t$. Therefore, with the ${\boldsymbol
\lambda}^{k+1}$ resulting from Step 3 of the IWFA, we have
\begin{equation}\label{Optimization: Auxiliary function maximizing}
\bar{C}_u({\boldsymbol \lambda}^{k+1})\geq \bar{C}_u({\boldsymbol
\lambda}^{k}) = {\Tilde C}_{u}({\boldsymbol \lambda}^{k}).
\end{equation}
From the concavity of ${\Tilde C}_{u}({\boldsymbol \lambda})$, it
can be shown that the following relation holds:
\begin{equation}\label{Optimization: Auxiliary function relation}
\bar{C}_u({\boldsymbol \lambda}^{k+1}) \leq {\tilde C}_u
\left(\frac{1}{N_t}{\boldsymbol
\lambda}^{k+1}+\frac{N_t-1}{N_t}{\boldsymbol \lambda}^{k}\right).
\end{equation}
Combining (\ref{Optimization: Auxiliary function maximizing}) and
(\ref{Optimization: Auxiliary function relation}) yields
\begin{equation}\label{Optimization: IWFA convergence relation}
{\Tilde C}_{u}({\boldsymbol \lambda}^{k}) \leq {\tilde C}_u
\left(\frac{1}{N_t}{\boldsymbol
\lambda}^{k+1}+\frac{N_t-1}{N_t}{\boldsymbol \lambda}^{k}\right).
\end{equation}
Therefore, after Step 5 of the IWFA, we have that ${\Tilde
C}_{u}({\boldsymbol \lambda}^{k+1}) \geq {\Tilde C}_{u}({\boldsymbol
\lambda}^{k})$. This, along with the fact that the problem
(\ref{Optimization: formula objective})-(\ref{Optimization: formula
constraints}) is convex, completes the proof.
\end{proof}

Notice that the relation (\ref{Optimization: IWFA convergence
relation}) suggests, mathematically, to update $\boldsymbol \lambda$
with $(\frac{1}{N_{t}}{\boldsymbol
\lambda}^{k+1}+\frac{N_{t}-1}{N_{t}}{\boldsymbol \lambda}^{k})$ in
the $k$-th iteration of the IWFA, whereas the KKT conditions
(\ref{Optimization: KKT conditions Simplified}) suggest a more
intuitive interpretation based on the water-filling principle. In
our proposed IWFA, we update $\boldsymbol \lambda$ with the
water-filling solution if the resulting ${\Tilde C}_{u}(\boldsymbol
\lambda)$ is increased. This allows very fast convergence, as we
demonstrate through simulations in the following section.
To guarantee the convergence, we use $(\frac{1}{N_{t}}{\boldsymbol
\lambda}^{k+1}+\frac{N_{t}-1}{N_{t}}{\boldsymbol \lambda}^{k})$ to
replace the water-filling solution when the resulting ${\Tilde
C}_{u}(\boldsymbol \lambda)$ is not increasing in each iteration.

\section{Simulations}\label{sec: Simulations}

In this section, we present numerical results to evaluate the
tightness of the capacity bound, and to demonstrate the efficiency
and performance of the proposed transmitter optimization approach
under SET.  We consider a MIMO system with five transmit and five
receive antennas, and present results for both the
jointly-correlated MIMO channel model and the Kronecker-correlation
model. For the jointly-correlated channel, we adopt the same channel
parameters as used in \cite{Veeravalli}, where $\bf D= 0$ and $\bf
\Omega$ has the following structure
\begin{equation}
{\bf \Omega}= \frac{25}{5.7} \left(
\begin{array}{ccccc} 0.1 & 0& 1& 0 & 0 \\ 0& 0.1 &1 & 0& 0 \\ 0& 0& 1 &0 &0
\\ 0& 0 &1 &0.25& 0\\0&0&1&0 &0.25
\end{array}\right).
\end{equation}
For the Kronecker channel, we adopt the constant-correlation model
for constructing the transmit and receive correlation matrices
\cite{Shin}. An $N \times N$ constant-correlation matrix is given by
\begin{equation}\label{eq:constCorrModel}
{\bf{\Theta }}_N \left( \alpha  \right) = \alpha {\bf{1 }}_{N \times
N} + \left( {1 - \alpha } \right){\bf{I}}_N,
\end{equation}
where $ \alpha \in \left[ {0,1} \right]$ is the correlation
coefficient. We set the transmit and receive correlation
coefficients to be $\alpha_t=0.4$ and $\alpha_r=0.6$ respectively.

Fig. \ref{fig:fig1} compares our closed-form ergodic mutual
information upper bound (\ref{capacity upbound}) with Monte-Carlo
simulated exact curves based on (\ref{mutual information}), for the
case $\boldsymbol \lambda=\bf 1$ (equal-power allocation). Results
are shown for both the jointly-correlated channel and the Kronecker
channel, with the above settings. We see that the upper bound is
rather tight for both channel models, especially for low to moderate
SNRs (eg.\ $< 8$ dB). Moreover, we see that the bound for the
Kronecker model is slightly tighter than for the more general
jointly-correlated model. Interestingly, we will show that, despite
this difference in tightness, the low-complexity power allocation
policies derived based on these bounds perform near-optimally for
both the Kronecker and jointly-correlated channel models.

Fig. \ref{fig:fig2} and Fig. \ref{fig:fig3} present the ergodic
mutual information achieved by the SET approach employing the
proposed IWFA (derived based on our closed-form upper bound), in the
jointly-correlated and Kronecker channel scenarios, respectively.
For comparison, the exact ergodic capacity curves are also shown,
which were obtained by numerically evaluating
(\ref{capacity_optimized}) using a constrained optimization function
of the Matlab optimization toolbox. The ergodic mutual information
achieved with equal power allocation (\ref{eq:EqualPower}) and
beamforming (\ref{eq:beamforming}) are also shown for further
comparison.  We clearly see that, for both channel models, the
proposed SET approach performs near-optimally, suffering
\emph{almost negligible loss} compared with the true channel
capacity. Furthermore, we see that equal power allocation and
beamforming are optimal in the high and low SNR regimes,
respectively, which agrees with our analytical conclusions put forth
in Section \ref{sec:OPA_Asymptotic}. The capacity upper bound curve
is also shown on the figures, and once again is seen to be tight.

Fig. \ref{fig:fig4} and Fig. \ref{fig:fig5} demonstrate the
convergence of the proposed IWFA for optimal power allocation in the
jointly-correlated and Kronecker channel scenarios, respectively.
Here, the SNR $\rho$ was set to 10 dB, and in all cases the
algorithm was initialized using ${\boldsymbol \lambda}^0 =
\mathbf{1}$. These figures show the evolution of the eigenvalues
$\lambda_i$, $i=1,\ldots,5$, and the capacity bound ${\Tilde
C}_u(\boldsymbol \lambda)$ for each iteration. From these results,
we see that the proposed IWFA converges after only a few iterations,
with the first iteration achieving near-optimal performance in all
cases.

\section{Conclusions}\label{sec:conclusion}
We have investigated statistical eigenmode transmission over a
general jointly-correlated MIMO channel. For this channel, we
derived a tight closed-form upper bound for the ergodic capacity,
which reveals a simple and interesting relationship in terms of
matrix permanents of the eigenvalue coupling matrix, and embraces
many existing results in the literature as special cases. Based on
this expression, we proposed and investigated new power allocation
policies in the framework of convex optimization. In particular, we
obtained necessary and sufficient optimality conditions, and
developed an efficient iterative water-filling algorithm with
guaranteed convergence. The tightness of the capacity bound and the
performance of our novel low-complexity transmitter optimization
approach was confirmed through simulations. Our approach was shown
to suffer near-negligible loss compared with the ergodic capacity of
the jointly-correlated MIMO channel.

\appendices

\section{Proof of Lemma \ref{lemma: Per:Polynomial expansion}}
\label{Appendix: Proof of lemma: Per:Polynomial expansion}

Let $\underline{\bf I}= \left[ {\bf I}_M \, {\bf 0}_{M\times
N}\right]$ and $\underline{\bf A}= \left[ {\bf 0}_{M\times M}\, {\bf
A} \right]$. From the definition of the permanents, we have
\begin{align}
\Per (\left[{\bf I}_{M} \, {\bf A}\right])
& = \Per (\bf{\underline{A}} + \bf{\underline{I}}) \nonumber\\
& = \sum \limits_{\hat{\beta}_M \in \mathcal{S}_{M+N}^{M}} \prod
\limits_{m=1}^{M} \left( \underline{a}_{m,\beta_m}+\underline{i}_{m,\beta_m}\right),\label{eq:74}
\end{align}
where $\underline{i}_{m,n}$ and $\underline{a}_{m,n}$ denote the
$(m,n)$-th elements of $\underline{\bf I}$ and $\underline{\bf A}$
respectively. Note that the following identity holds:
\begin{equation}
\prod \limits_{m=1}^{M} \left(x_m+y_m\right)= \sum \limits_{k=0}^{M}
\sum \limits_{\hat{\alpha}_k \in \mathcal{S}_{M}^{(k)}} \prod
\limits_{m=1}^{k} x_{\alpha_m} \prod \limits_{m=1}^{M-k}
y_{\alpha_m^{'}},
\end{equation}
where $(\alpha_1^{'}, \alpha_2^{'}, ..., \alpha_{M-k}^{'}) \in
\mathcal{S}_{M}^{M-k}$ is the sequence complementary to
$\hat{\alpha}_k$ in $\{1, 2, ..., M \}$. Hence
\begin{equation}\label{eq:76}
\Per (\left[{\bf I}_{M} \, {\bf A}\right])
= \sum \limits_{k=0}^{M}
\sum \limits_{\hat{\alpha}_k \in \mathcal{S}_{M}^{(k)}} \left( \sum
\limits_{\hat{\beta}_M \in \mathcal{S}_{M+N}^{M}} \prod
\limits_{m=1}^{k} \underline{a}_{\alpha_m, \beta_{\alpha_{m}}} \prod
\limits_{m=1}^{M-k} \underline{i}_{\alpha_{m}^{'},
\beta_{\alpha_{m}^{'}}} \right).
\end{equation}
It can be seen that $ \underline{i}_{\alpha_m^{'},
\beta_{\alpha_{m}^{'}}}= \delta
(\beta_{\alpha_{m}^{'}}-\alpha_m^{'})$, where $\delta(\cdot)$ is the
Kronecker delta operator, and $\prod _{m=1}^{k}
\underline{a}_{\alpha_m, \beta_{\alpha_{m}}}\neq 0$ only if
$\beta_{\alpha_{m}}>M$ and $k\leq {\rm min}(M,N)$. Therefore, we
have
\begin{align}
\Per (\left[{\bf I}_{M} \, {\bf A}\right]) & =\sum
\limits_{k=0}^{{\rm min}(M,N)} \sum \limits_{\hat{\alpha}_k \in
\mathcal{S}_{M}^{(k)}} \sum \limits_{\hat{\beta}_k \in
\mathcal{S}_{N}^{k}} \prod \limits_{m=1}^{k}
\underline{a}_{\alpha_m, M+\beta_{m}} \nonumber \\
&= \sum \limits_{k=0}^{{\rm min}(M,N)} \sum
\limits_{\hat{\alpha}_k \in \mathcal{S}_{M}^{(k)}} \sum
\limits_{\hat{\beta}_k \in \mathcal{S}_{N}^{(k)}} \Per \left({\bf
A}^{\hat{\alpha}_k}_{{\hat{\beta}_k}} \right),
\end{align}
where  $\Per \left({\bf A}^{\hat{\alpha}_k}_{{\hat{\beta}_k}}
\right)=1$ when $k=0$. Using (\ref{Per: property 7}), we have
\begin{align}
\Per (\left[{\bf I}_{M} \, {\bf A}\right]) &= \sum \limits_{k=0}^{\min(M,N)} \sum
\limits_{{\hat{\alpha}}_k\in \mathcal{S}_M^{(k)}} \Per ({\bf A}^{\hat\alpha_k }) \nonumber \\
&=  \sum \limits_{k=0}^{\min(M,N)} \sum \limits_{{\hat{\beta}}_k\in \mathcal{S}_N^{(k)}} \Per({\bf
A}_{\hat\beta_k }).
\end{align}
Through a similar procedure, one can obtain that
\begin{align}
\Per (\left[{\bf I}_{N} \, {\bf A}^{T} \right])& =  \sum \limits_{k=0}^{\min(M,N)} \sum
\limits_{{\hat{\alpha}}_k\in \mathcal{S}_M^{(k)}} \Per ({\bf A}^{\hat\alpha_k }) \nonumber \\
&=  \sum \limits_{k=0}^{\min(M,N)} \sum \limits_{{\hat{\beta}}_k\in \mathcal{S}_N^{(k)}} \Per({\bf
A}_{\hat\beta_k }).
\end{align}
This completes the proof.

\section{Proof of Lemma \ref{Lemma: E_det_prod}}\label{Appendix: Proof of Lemma: E_det_prod}

From the definition of the determinant, we have
\begin{equation}\label{Lemma1:eq1}
\mathbb{E}\left\{ \det\left( \bf X \right) \det\left( {\bf X}^{H}
\right) \right\} = \sum\limits_{\hat{\alpha}_{N}\in \mathcal
{S}_N^N}\sum\limits_{\hat{\beta}_{N}\in \mathcal {S}_N^N}
\!\!\! (-1)^{\sigma (\hat{\alpha}_{N})}(-1)^{\sigma
(\hat{\beta}_{N})}\mathbb{E}\left\{\prod\limits_{i=1}^{N}
x_{i,\alpha_{i}} x_{i,\beta_{i}}^{*}\right\},
\end{equation}
where ${\sigma (\hat{\alpha}_{N})}$ denotes the number of inversions
in the permutation $\hat{\alpha}_{N}$ from the normal order
$1,2,\ldots ,N$, and $x_{i,j}$ is the $(i,j)$-th element of $\bf X$.
Since the rows of $\bf X$ are independent, we have
\begin{equation}\label{Lemma1:eq2}
\mathbb{E}\left\{\prod\limits_{i=1}^{N} x_{i,\alpha_{i}}
x_{i,\beta_{i}}^{*}\right\}=\prod\limits_{i=1}^{N}\mathbb{E}\left\{
x_{i,\alpha_{i}} x_{i,\beta_{i}}^{*}\right\}.
\end{equation}
Since the elements in each row are independent and there is only one
possible non-zero mean element in each row, we have
\begin{equation}\label{Lemma1:eq3}
\mathbb{E}\left\{ x_{i,\alpha_{i}}
x_{i,\beta_{i}}^{*}\right\}=\xi_{i,\alpha_i}\delta(\beta_i-\alpha_i),
\end{equation}
where $\xi_{i,j}$ is the $(i,j)$-th element of $\bf \Xi $.
Substituting (\ref{Lemma1:eq3}) into (\ref{Lemma1:eq2}) and then
into (\ref{Lemma1:eq1}) yield
\begin{equation}\label{Lemma1:eq4}
\mathbb{E}\left\{ \det\left( \bf X \right) \det\left( {\bf X}^{H}
\right) \right\}= \sum\limits_{\hat{\alpha}_{N}\in \mathcal
{S}_N^N}\prod\limits_{i=1}^{N} \xi_{i,\alpha_{i}}=\Per \left({\bf
\Xi} \right).
\end{equation}
This completes the proof.

\section{Proof of the concavity of $f(\boldsymbol \lambda)=\log \Per({\bf A}\diag(\boldsymbol
\lambda))$ in several cases }\label{Appendix: Proof of
Lemma:Concavity}

{\emph Case 1}: $M\geq N$. In this case, we have that $f(\boldsymbol
\lambda)=\log \Per({\bf A})+\log\det(\diag(\boldsymbol \lambda))$.
The concavity of $f(\boldsymbol \lambda)$ comes from that of
$\log\det(\diag(\boldsymbol \lambda))$ \cite{Boyd}.

{\emph Case 2}: $M=1$ and $N>1$. In this case, $\bf A$ is a row
vector, and we have that $f(\boldsymbol \lambda)=\log({\bf
A}\boldsymbol \lambda)$. The concavity of $f(\boldsymbol \lambda)$
comes from that of the log function.

{\emph Case 3}: $M=2$ and $N>2$. In this case, we will first show
that the following inequality holds:
\begin{equation}\label{Inequality 1 for M=2}
\frac{g_2(\boldsymbol \lambda_1+\boldsymbol
\lambda_2)}{g_1(\boldsymbol \lambda_1+\boldsymbol
\lambda_2)}\geq\frac{g_2(\boldsymbol \lambda_1)}{g_1(\boldsymbol
\lambda_1)}+\frac{g_2(\boldsymbol \lambda_2)}{g_1(\boldsymbol
\lambda_2)},
\end{equation}
where $g_1(\boldsymbol \lambda)= {\bf 1}_{1\times 2}{\bf
A}\diag(\boldsymbol \lambda)$ and $g_2(\boldsymbol \lambda)=
\Per({\bf A}\diag(\boldsymbol \lambda))$. Then we will prove the
concavity of $f(\boldsymbol \lambda)$ from (\ref{Inequality 1 for
M=2}).

Since $g_1(\boldsymbol \lambda)$ and $g_2(\boldsymbol \lambda)$ are
positive on $\mathcal{D}^N$, the inequality (\ref{Inequality 1 for
M=2}) holds if and only if the following inequality does:
\begin{equation}\label{Inequality 2 for M=2}
g(\boldsymbol \lambda_1, \boldsymbol \lambda_2)
= g_2(\boldsymbol \lambda_1+\boldsymbol \lambda_2)g_1(\boldsymbol
\lambda_1)g_1(\boldsymbol \lambda_2)
- g_2(\boldsymbol \lambda_1)g_1(\boldsymbol \lambda_1+\boldsymbol
\lambda_2)g_1(\boldsymbol \lambda_2)
- g_2(\boldsymbol \lambda_2)g_1(\boldsymbol \lambda_1+\boldsymbol
\lambda_2)g_1(\boldsymbol \lambda_1)
\geq 0.
\end{equation}
Let ${\bf A}=[{\bf a}_1^T\;{\bf a}_2^{T}]^T$. Then we have that
$g_1(\boldsymbol \lambda)={\bf a}_1^T\boldsymbol \lambda+{\bf
a}_2^T\boldsymbol \lambda$ and $g_2(\boldsymbol \lambda)={\bf
a}_1^T\boldsymbol \lambda {\bf a}_2^T\boldsymbol \lambda-
\boldsymbol \lambda ^T\diag({\bf a}_1\odot {\bf a}_2)\boldsymbol
\lambda$. By substituting these expressions into $g(\boldsymbol
\lambda_1, \boldsymbol \lambda_2)$, we can obtain
\begin{equation}\label{g-fun for M=2}
g(\boldsymbol \lambda_1, \boldsymbol \lambda_2)
 = ({\bf a}_1^T\boldsymbol \lambda_1 {\bf a}_2^T\boldsymbol \lambda_2 -{\bf
a}_1^T\boldsymbol \lambda_2 {\bf a}_2^T\boldsymbol
\lambda_1)^2
  + (g_1(\boldsymbol \lambda_2)\boldsymbol
\lambda_1-g_1(\boldsymbol \lambda_1)\boldsymbol \lambda_2)
^T\diag({\bf a}_1\odot {\bf a}_2)
  (g_1(\boldsymbol
\lambda_2)\boldsymbol \lambda_1-g_1(\boldsymbol
\lambda_1)\boldsymbol \lambda_2).
\end{equation}
Therefore we achieve (\ref{Inequality 2 for M=2}) and then
(\ref{Inequality 1 for M=2}). From (\ref{Inequality 1 for M=2}), we
have
\begin{equation}\label{Inequality 3 for M=2}
\frac{g_2(\theta \boldsymbol \lambda_1+(1-\theta )\boldsymbol
\lambda_2)}{g_1(\theta \boldsymbol \lambda_1+(1-\theta )\boldsymbol
\lambda_2)}\geq \theta  \frac{g_2(\boldsymbol
\lambda_1)}{g_1(\boldsymbol \lambda_1)}+(1-\theta )
\frac{g_2(\boldsymbol \lambda_2)}{g_1(\boldsymbol \lambda_2)},
\end{equation}
where $0\leq \theta \leq 1$. Taking logarithm on both sides and
using the concavity of the log function yields
\begin{equation}\label{Inequality 4 for M=2}
\begin{split}
& f(\theta \boldsymbol \lambda_1+(1-\theta )\boldsymbol
\lambda_2)-\theta f(\boldsymbol \lambda_1)-(1-\theta )f(\boldsymbol
\lambda_2) \\
& \quad \geq \log (g_1(\theta \boldsymbol
\lambda_1+(1-\theta )\boldsymbol \lambda_2))+\theta \log
(g_1(\boldsymbol \lambda_1)) - (1-\theta ) \log (g_1(\boldsymbol \lambda_2)) \geq 0.
\end{split}
\end{equation}
This completes the proof of the concavity of $f(\boldsymbol \lambda)$.

{\emph Case 4}: $\bf A$ is of rank one. Let ${\bf A}={\bf a}{\bf
b}^T$, where $\bf a$ and $\bf b$ are vectors of $M$ and $N$ elements
respectively. In this case, we have
\begin{equation}\label{Concavity rank one}
\begin{split}
f(\boldsymbol \lambda)
& = \log \Per({\bf 1}_{M\times N} \diag({\bf b}
\odot \boldsymbol \lambda))+\log\det(\diag(\bf a)) \\
& = \log \sum \limits_{{\hat{\alpha}_{M}}\in \mathcal {S}_{N}^{(M)}} {
\Per \left(({\bf 1}_{M\times N} \diag({\bf b} \odot \boldsymbol
\lambda))_{\hat{\alpha}_{N}}\right)}  + \log\det(\diag(\bf a)) \\
& = \log E_M({\bf b} \odot \boldsymbol
\lambda)+\log(M!)+\log\det(\diag(\bf a)),
\end{split}
\end{equation}
where the function $E_M(\boldsymbol \lambda)$ is the $M$-th
elementary symmetric function defined as \cite{Marshall}
\begin{equation} \label{elementary symetric function}
E_M(\boldsymbol \lambda)=\sum\limits_{{\hat{\alpha}_{M}}\in \mathcal
{S}_{N}^{(M)}}\prod_{i=1}^{M}\lambda_{\alpha_i}.
\end{equation}
Since $E_M(\boldsymbol \lambda)$ is logarithmically concave, we
obtain from (\ref{Concavity rank one}) that $f(\boldsymbol \lambda)$
is concave.

\section{Proof of Lemma \ref{Lemma: Per: poly in z}}\label{Appendix: Proof of Lemma: Per: poly in z}

We consider the case with $N_r \leq N_t$. The proof for the case
with $N_r > N_t$ is similar. From the definition of the permanents,
we have
\begin{equation} \label{Lemma Per Poly: eq1}
 \Per ({\hat{\bf \Omega
}(z)}) =  \sum \limits_{{\hat{\alpha}_{N_r}}\in \mathcal
{S}_{N_t}^{N_r}} \prod \limits_{i=1}^{N_r} (1+\hat{\omega}
_{i,{\alpha}_{i}}z),
\end{equation}
where $\hat{\omega} _{i,j}$ represents the $(i,j)$-th element of
$\hat{ \bf
 \Omega }$. For each product term in the above expression, the
following relation holds:
\begin{equation} \label{Lemma Per Poly: eq2}
 \prod \limits_{i=1}^{N_r} (1+\hat{\omega}
_{i,{\alpha}_{i}}z)= \sum \limits_{k=0}^{N_r} z^{k}\sum
\limits_{{\hat{\beta}_{k}}\in \mathcal {S}_{N_r}^{(k)}}\prod
\limits_{i=1}^{k} \hat{\omega} _{\beta_{i}, \alpha_{\beta_{i}}}.
\end{equation}
Substituting (\ref{Lemma Per Poly: eq2}) into (\ref{Lemma Per Poly:
eq1}) yields
\begin{align} \label{Lemma Per Poly: eq3}
\Per ({\hat{\bf \Omega}(z)}) &= \sum \limits_{k=0}^{N_r} z^{k}\sum
\limits_{{\hat{\beta}_{k}}\in \mathcal {S}_{N_r}^{(k)}}\sum
\limits_{{\hat{\alpha}_{N_r}}\in \mathcal {S}_{N_t}^{N_r}}\prod
\limits_{i=1}^{k} \hat{\omega} _{\beta_{i}, \alpha_{\beta_{i}}} \nonumber \\
&= \sum \limits_{k=0}^{N_r} z^{k}\frac{(N_t-k)!}{(N_t-N_r)!}\sum
\limits_{{\hat{\beta}_{k}}\in \mathcal {S}_{N_r}^{(k)}}\sum
\limits_{{\hat{\alpha}_{k}}\in \mathcal {S}_{N_t}^{k}}\prod
\limits_{i=1}^{k} \hat{\omega }_{\beta_{i}, \alpha_{i}} \nonumber \\
&= \sum \limits_{k=0}^{N_r} z^{k}\frac{(N_t-k)!}{(N_t-N_r)!}\sum
\limits_{{\hat{\beta}_{k}}\in \mathcal {S}_{N_r}^{(k)}}
\Per(\hat{\bf \Omega }^{\hat{\beta}_{k}}).
\end{align}
From Lemma \ref{lemma: Per:Polynomial expansion}, we have the
expansion of $\underline{\Per}( \hat{\bf \Omega })$. By comparing
the resulting expansion of $\underline{\Per}( \hat{\bf \Omega })$
with (\ref{Lemma Per Poly: eq3}), we complete the proof.



\newpage
\begin{figure*}
\centering
\begin{psfrags}%
\psfragscanon%
%
\small{
\psfrag{s05}[t][t]{\color[rgb]{0,0,0}\setlength{\tabcolsep}{0pt}\begin{tabular}{c}$ N(=N_t=N_r)$\end{tabular}}%
\psfrag{s06}[b][b]{\color[rgb]{0,0,0}\setlength{\tabcolsep}{0pt}\begin{tabular}{c}Multiplication Number\end{tabular}}%
\psfrag{s10}[l][l]{\color[rgb]{0,0,0}Polynomial+Rysers formula}%
\psfrag{s11}[l][l]{\color[rgb]{0,0,0}Definition}%
\psfrag{s12}[l][l]{\color[rgb]{0,0,0}Laplace expansion}%
\psfrag{s13}[l][l]{\color[rgb]{0,0,0}Rysers formula}%
\psfrag{s14}[l][l]{\color[rgb]{0,0,0}Polynomial+Definition}%
\psfrag{s15}[l][l]{\color[rgb]{0,0,0}Polynomial+Laplace expansion}%
\psfrag{s16}[l][l]{\color[rgb]{0,0,0}Polynomial+Rysers formula}%
%
\psfrag{x01}[t][t]{2}%
\psfrag{x02}[t][t]{3}%
\psfrag{x03}[t][t]{4}%
\psfrag{x04}[t][t]{5}%
\psfrag{x05}[t][t]{6}%
\psfrag{x06}[t][t]{7}%
\psfrag{x07}[t][t]{8}%
%
\psfrag{v01}[r][r]{$10^{0}$}%
\psfrag{v02}[r][r]{$10^{2}$}%
\psfrag{v03}[r][r]{$10^{4}$}%
\psfrag{v04}[r][r]{$10^{6}$}%
\psfrag{v05}[r][r]{$10^{8}$}%
\psfrag{v06}[r][r]{$10^{10}$}%
%
\resizebox{12cm}{!}{\includegraphics{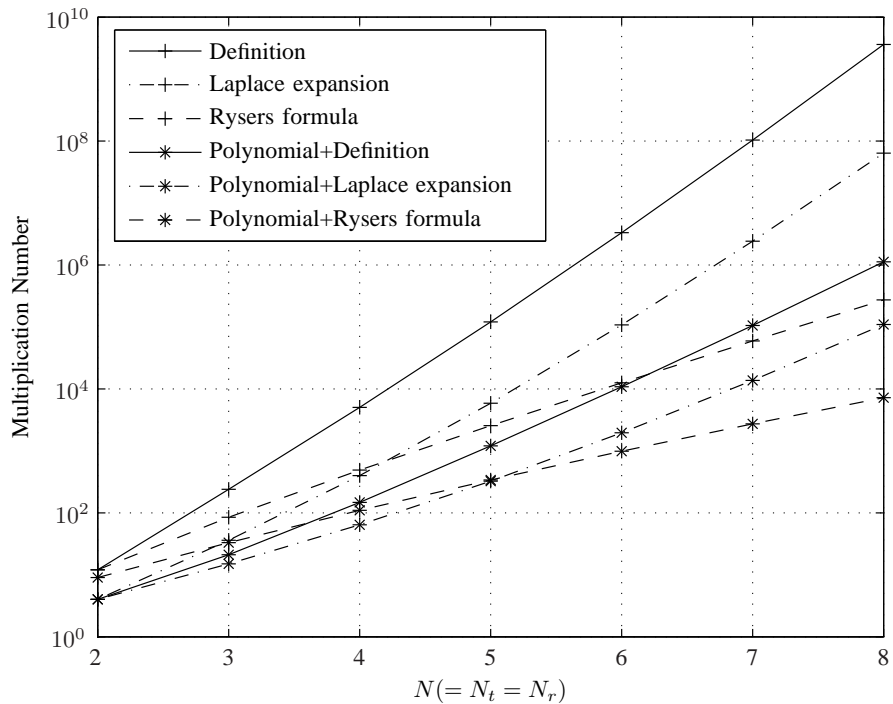}}%
}
\end{psfrags}%
\caption{Comparison of the number of required multiplications for
calculating ${\Tilde C}_{u}({\boldsymbol \lambda})$ with the
polynomial-based algorithms and the direct computation algorithms.}
\label{fig:fig0}
\end{figure*}
%
\begin{figure*}
\centering
\begin{psfrags}%
\psfragscanon%
%
\small{
\psfrag{s05}[t][t]{\color[rgb]{0,0,0}\setlength{\tabcolsep}{0pt}\begin{tabular}{c}SNR (in dB)\end{tabular}}%
\psfrag{s06}[b][b]{\color[rgb]{0,0,0}\setlength{\tabcolsep}{0pt}\begin{tabular}{c}Ergodic Capacity (in bps/Hz)\end{tabular}}%
\psfrag{s10}[l][l]{\color[rgb]{0,0,0}Upper bound for Kronecker channel}%
\psfrag{s11}[l][l]{\color[rgb]{0,0,0}Jointly correlated channel}%
\psfrag{s12}[l][l]{\color[rgb]{0,0,0}    Upper bound for jointly correlated channel}%
\psfrag{s13}[l][l]{\color[rgb]{0,0,0}Kronecker channel}%
\psfrag{s14}[l][l]{\color[rgb]{0,0,0}Upper bound for Kronecker channel}%
%
\psfrag{x01}[t][t]{2}%
\psfrag{x02}[t][t]{4}%
\psfrag{x03}[t][t]{6}%
\psfrag{x04}[t][t]{8}%
\psfrag{x05}[t][t]{10}%
\psfrag{x06}[t][t]{12}%
\psfrag{x07}[t][t]{14}%
\psfrag{x08}[t][t]{16}%
%
\psfrag{v01}[r][r]{2}%
\psfrag{v02}[r][r]{4}%
\psfrag{v03}[r][r]{6}%
\psfrag{v04}[r][r]{8}%
\psfrag{v05}[r][r]{10}%
\psfrag{v06}[r][r]{12}%
\psfrag{v07}[r][r]{14}%
\psfrag{v08}[r][r]{16}%
\psfrag{v09}[r][r]{18}%
%
\resizebox{12cm}{!}{\includegraphics{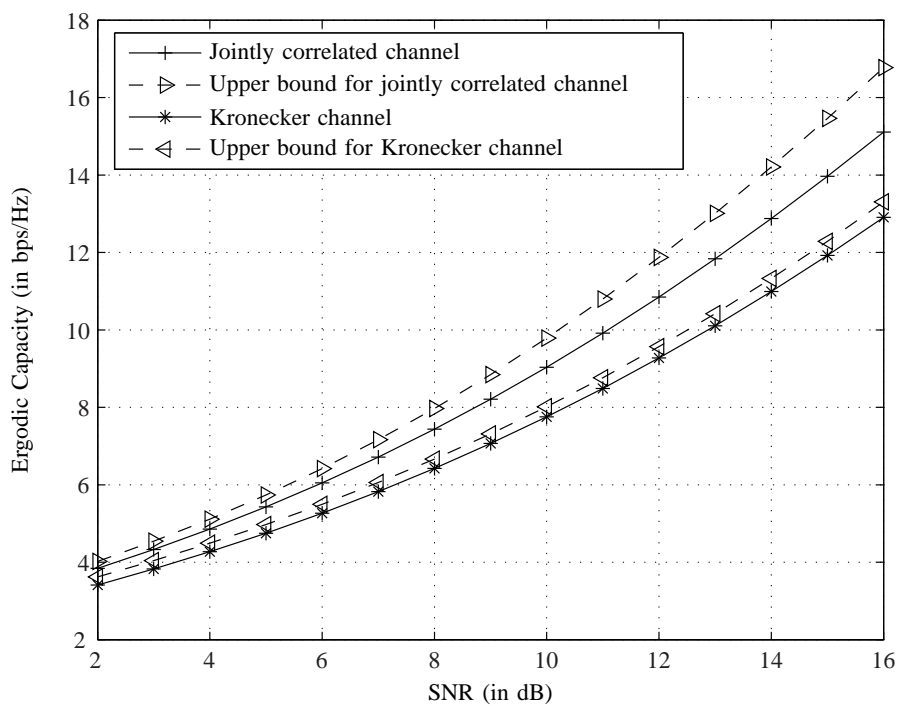}}%
}
\end{psfrags}%
\caption{Comparison of the exact ergodic mutual information and the
mutual information upper bound.  Results are shown for the
jointly-correlated channel model and the Kronecker model, with
equal-power allocation (i.e.\ ${\boldsymbol \lambda} = {\bf 1}$).}
\label{fig:fig1}
\end{figure*}

\begin{figure*}
\centering
\begin{psfrags}%
\psfragscanon%
%
\small{
\psfrag{s05}[t][t]{\color[rgb]{0,0,0}\setlength{\tabcolsep}{0pt}\begin{tabular}{c}SNR (in dB)\end{tabular}}%
\psfrag{s06}[b][b]{\color[rgb]{0,0,0}\setlength{\tabcolsep}{0pt}\begin{tabular}{c}Ergodic Capacity (in bps/Hz)\end{tabular}}%
\psfrag{s10}[l][l]{\color[rgb]{0,0,0}Beamforming}%
\psfrag{s11}[l][l]{\color[rgb]{0,0,0}   Numerical solution of (14)}%
\psfrag{s12}[l][l]{\color[rgb]{0,0,0}IWFA}%
\psfrag{s13}[l][l]{\color[rgb]{0,0,0}Upper bound}%
\psfrag{s14}[l][l]{\color[rgb]{0,0,0}Equal power allocation}%
\psfrag{s15}[l][l]{\color[rgb]{0,0,0}Beamforming}%
%
\psfrag{x01}[t][t]{2}%
\psfrag{x02}[t][t]{4}%
\psfrag{x03}[t][t]{6}%
\psfrag{x04}[t][t]{8}%
\psfrag{x05}[t][t]{10}%
\psfrag{x06}[t][t]{12}%
\psfrag{x07}[t][t]{14}%
\psfrag{x08}[t][t]{16}%
%
\psfrag{v01}[r][r]{2}%
\psfrag{v02}[r][r]{4}%
\psfrag{v03}[r][r]{6}%
\psfrag{v04}[r][r]{8}%
\psfrag{v05}[r][r]{10}%
\psfrag{v06}[r][r]{12}%
\psfrag{v07}[r][r]{14}%
\psfrag{v08}[r][r]{16}%
\psfrag{v09}[r][r]{18}%
%
\resizebox{12cm}{!}{\includegraphics{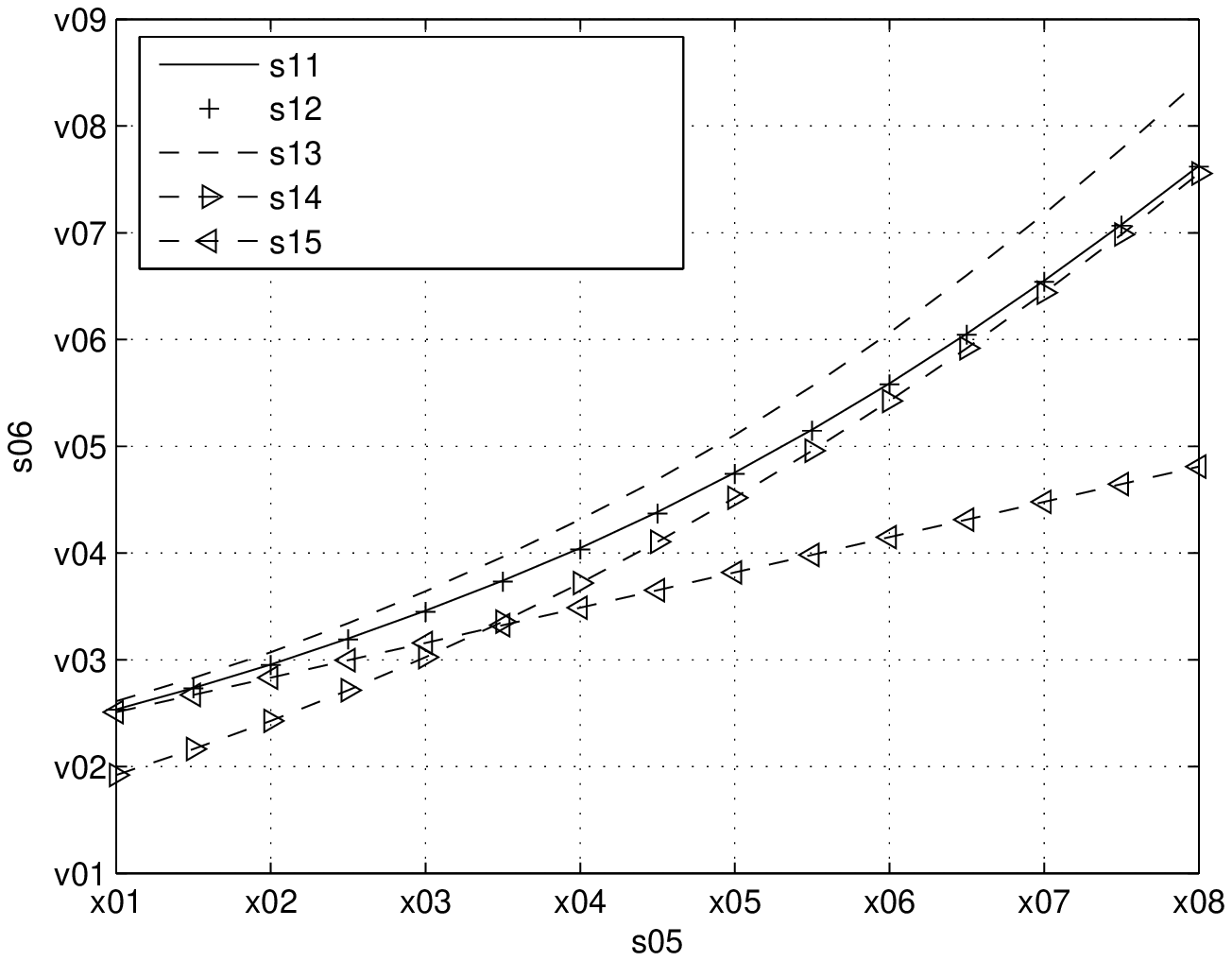}}%
}
\end{psfrags}%
\caption{Comparison of the ergodic capacity of the
jointly-correlated MIMO channel achieved by numerically solving
(\ref{capacity_optimized}), and our proposed iterative water-filling
algorithm under SET.  The capacity upper bound and the information
rates achieved by equal power allocation and beamforming are also
shown. } \label{fig:fig2}
\end{figure*}

\begin{figure*}
\centering
\begin{psfrags}%
\psfragscanon%
%
\small{
\psfrag{s05}[t][t]{\color[rgb]{0,0,0}\setlength{\tabcolsep}{0pt}\begin{tabular}{c}SNR (in dB)\end{tabular}}%
\psfrag{s06}[b][b]{\color[rgb]{0,0,0}\setlength{\tabcolsep}{0pt}\begin{tabular}{c}Ergodic Capacity (in bps/Hz)\end{tabular}}%
\psfrag{s10}[l][l]{\color[rgb]{0,0,0}Beamforming}%
\psfrag{s11}[l][l]{\color[rgb]{0,0,0}   Numerical solution of (14)}%
\psfrag{s12}[l][l]{\color[rgb]{0,0,0}IWFA}%
\psfrag{s13}[l][l]{\color[rgb]{0,0,0}Upper bound}%
\psfrag{s14}[l][l]{\color[rgb]{0,0,0}Equal power allocation}%
\psfrag{s15}[l][l]{\color[rgb]{0,0,0}Beamforming}%
%
\psfrag{x01}[t][t]{2}%
\psfrag{x02}[t][t]{4}%
\psfrag{x03}[t][t]{6}%
\psfrag{x04}[t][t]{8}%
\psfrag{x05}[t][t]{10}%
\psfrag{x06}[t][t]{12}%
\psfrag{x07}[t][t]{14}%
\psfrag{x08}[t][t]{16}%
%
\psfrag{v01}[r][r]{2}%
\psfrag{v02}[r][r]{4}%
\psfrag{v03}[r][r]{6}%
\psfrag{v04}[r][r]{8}%
\psfrag{v05}[r][r]{10}%
\psfrag{v06}[r][r]{12}%
\psfrag{v07}[r][r]{14}%
%
\resizebox{12cm}{!}{\includegraphics{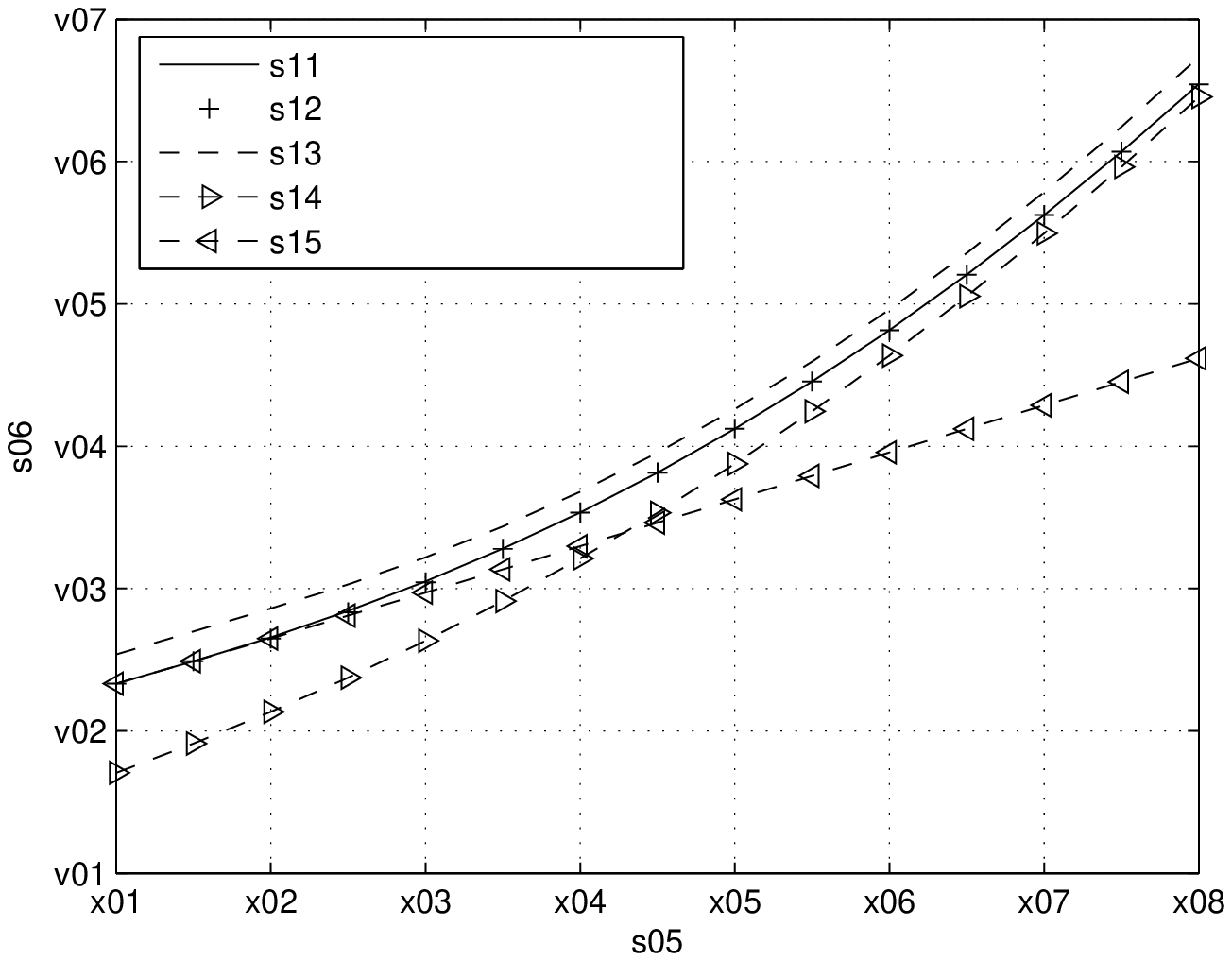}}%
}
\end{psfrags}%
\caption{Comparison of the ergodic capacity of the Kronecker MIMO
channel achieved by numerically solving (\ref{capacity_optimized}),
and our proposed iterative water-filling algorithm under SET.  The
capacity upper bound and the information rates achieved by equal
power allocation and beamforming are also shown.} \label{fig:fig3}
\end{figure*}

\begin{figure*}
\centering
\begin{psfrags}%
\psfragscanon%
%
\small{
\psfrag{s07}[t][t]{\color[rgb]{0,0,0}\setlength{\tabcolsep}{0pt}\begin{tabular}{c}Iteration ($ k $)\end{tabular}}%
\psfrag{s08}[b][b]{\color[rgb]{0,0,0}\setlength{\tabcolsep}{0pt}\begin{tabular}{c}$\lambda_{i}^{k}$\end{tabular}}%
\psfrag{s12}[t][t]{\color[rgb]{0,0,0}\setlength{\tabcolsep}{0pt}\begin{tabular}{c}Iteration ($k$)\end{tabular}}%
\psfrag{s13}[b][b]{\color[rgb]{0,0,0}\setlength{\tabcolsep}{0pt}\begin{tabular}{c}$\tilde{C}_{u}({\bf \lambda})$\end{tabular}}%
\psfrag{s14}[l][l]{\color[rgb]{0,0,0}$i = 5$}%
\psfrag{s15}[l][l]{\color[rgb]{0,0,0}$ i = 1$}%
\psfrag{s16}[l][l]{\color[rgb]{0,0,0}$i = 2$}%
\psfrag{s17}[l][l]{\color[rgb]{0,0,0}$ i = 3$}%
\psfrag{s18}[l][l]{\color[rgb]{0,0,0}$ i = 4$}%
\psfrag{s19}[l][l]{\color[rgb]{0,0,0}$i = 5$}%
%
\psfrag{x01}[t][t]{0}%
\psfrag{x02}[t][t]{1}%
\psfrag{x03}[t][t]{2}%
\psfrag{x04}[t][t]{3}%
\psfrag{x05}[t][t]{4}%
\psfrag{x06}[t][t]{5}%
\psfrag{x07}[t][t]{6}%
\psfrag{x08}[t][t]{0}%
\psfrag{x09}[t][t]{1}%
\psfrag{x10}[t][t]{2}%
\psfrag{x11}[t][t]{3}%
\psfrag{x12}[t][t]{4}%
\psfrag{x13}[t][t]{5}%
\psfrag{x14}[t][t]{6}%
%
\psfrag{v01}[r][r]{9.8}%
\psfrag{v02}[r][r]{10}%
\psfrag{v03}[r][r]{10.2}%
\psfrag{v04}[r][r]{10.4}%
\psfrag{v05}[r][r]{10.6}%
\psfrag{v06}[r][r]{0}%
\psfrag{v07}[r][r]{0.5}%
\psfrag{v08}[r][r]{1}%
\psfrag{v09}[r][r]{1.5}%
\psfrag{v10}[r][r]{2}%
%
\resizebox{12cm}{!}{\includegraphics{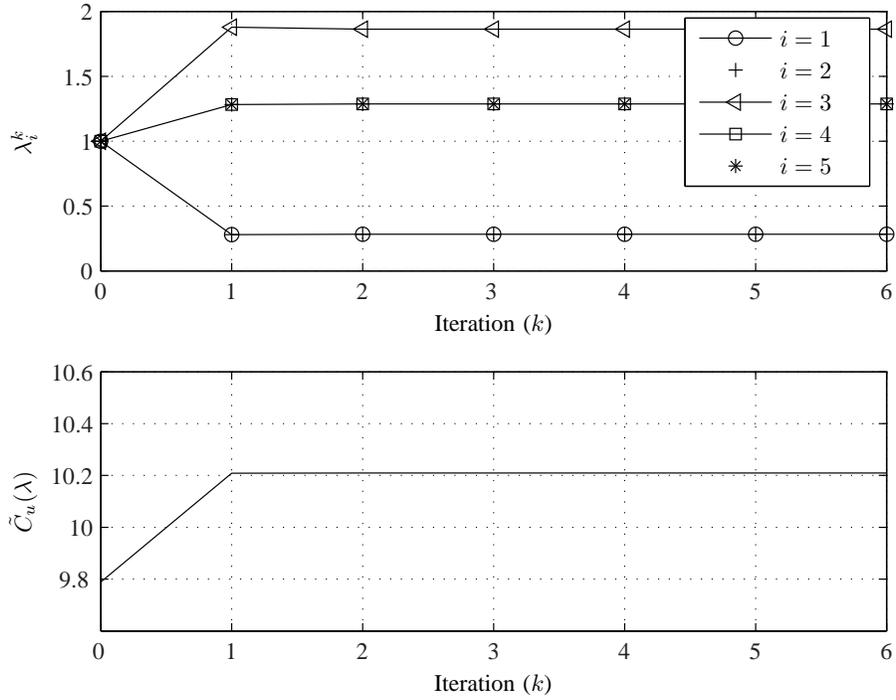}}%
}
\end{psfrags}%
\caption{Convergence of the iterative water-filling algorithm for
optimal power allocation in the jointly-correlated channel.  Results
are shown for SNR = 10 dB. }
 \label{fig:fig4}
\end{figure*}

\begin{figure*}
\centering
\begin{psfrags}%
\psfragscanon%
%
\small{
\psfrag{s07}[t][t]{\color[rgb]{0,0,0}\setlength{\tabcolsep}{0pt}\begin{tabular}{c}Iteration ($ k $)\end{tabular}}%
\psfrag{s08}[b][b]{\color[rgb]{0,0,0}\setlength{\tabcolsep}{0pt}\begin{tabular}{c}$\lambda_{i}^{k}$\end{tabular}}%
\psfrag{s12}[t][t]{\color[rgb]{0,0,0}\setlength{\tabcolsep}{0pt}\begin{tabular}{c}Iteration ($k$)\end{tabular}}%
\psfrag{s13}[b][b]{\color[rgb]{0,0,0}\setlength{\tabcolsep}{0pt}\begin{tabular}{c}$\tilde{C}_{u}({\bf \lambda})$\end{tabular}}%
\psfrag{s14}[l][l]{\color[rgb]{0,0,0}$i = 5$}%
\psfrag{s15}[l][l]{\color[rgb]{0,0,0}$ i = 1$}%
\psfrag{s16}[l][l]{\color[rgb]{0,0,0}$i = 2$}%
\psfrag{s17}[l][l]{\color[rgb]{0,0,0}$ i = 3$}%
\psfrag{s18}[l][l]{\color[rgb]{0,0,0}$ i = 4$}%
\psfrag{s19}[l][l]{\color[rgb]{0,0,0}$i = 5$}%
%
\psfrag{x01}[t][t]{0}%
\psfrag{x02}[t][t]{1}%
\psfrag{x03}[t][t]{2}%
\psfrag{x04}[t][t]{3}%
\psfrag{x05}[t][t]{4}%
\psfrag{x06}[t][t]{5}%
\psfrag{x07}[t][t]{6}%
\psfrag{x08}[t][t]{0}%
\psfrag{x09}[t][t]{1}%
\psfrag{x10}[t][t]{2}%
\psfrag{x11}[t][t]{3}%
\psfrag{x12}[t][t]{4}%
\psfrag{x13}[t][t]{5}%
\psfrag{x14}[t][t]{6}%
%
\psfrag{v01}[r][r]{8}%
\psfrag{v02}[r][r]{8.2}%
\psfrag{v03}[r][r]{8.4}%
\psfrag{v04}[r][r]{8.6}%
\psfrag{v05}[r][r]{0.5}%
\psfrag{v06}[r][r]{1}%
\psfrag{v07}[r][r]{1.5}%
\psfrag{v08}[r][r]{2}%
\psfrag{v09}[r][r]{2.5}%
\psfrag{v10}[r][r]{3}%
%
\resizebox{12cm}{!}{\includegraphics{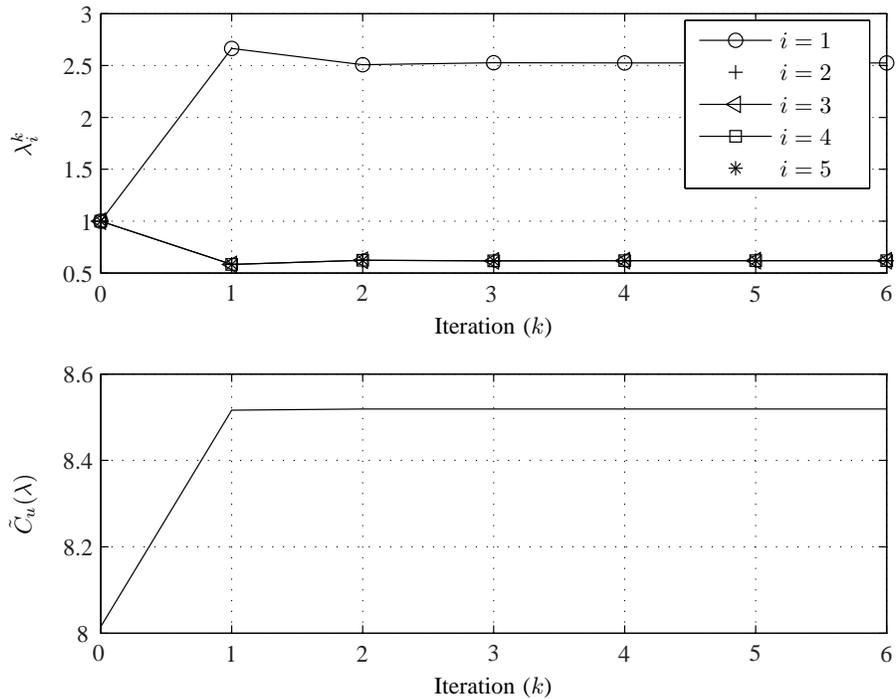}}%
}
\end{psfrags}%
\caption{Convergence of the iterative water-filling algorithm for
optimal power allocation in the Kronecker channel. Results are shown
for SNR=10 dB.} \label{fig:fig5}
\end{figure*}

\end{document}